\documentclass[11pt]{scrartcl}

\usepackage{ifthen}
\usepackage{comment}

\newboolean{short}
\setboolean{short}{false}

\newcommand{\onlyShort}[1]{\ifthenelse{\boolean{short}}{#1}{}}
\newcommand{\onlyLong}[1]{\ifthenelse{\boolean{short}}{}{#1}}

\usepackage[normalem]{ulem}
\usepackage{amssymb,amsmath,amsthm}

\usepackage[]{color}
\usepackage{comment}
\usepackage{paralist}
\usepackage[mathscr]{eucal}

\onlyLong{

\usepackage{algorithm}
\usepackage[noend]{algorithmic}
\usepackage{algorithmicext}
\usepackage{multirow}
\usepackage{tikz}
\usetikzlibrary{arrows,decorations.pathmorphing,backgrounds,positioning,fit}
\usepackage{calc}

\usepackage{exscale}
\usepackage{xspace}
\usepackage{setspace}
\usepackage{fancyhdr}
\usepackage{cite}
\usepackage{subfig}
\usepackage{caption}    %
\usepackage{layout}
\usepackage{multicol}

}

\onlyShort{
\def\Section{Sect.}
}
\onlyLong{
\def\Section{Section}
}

\def\true{\mbox{\sc true}}
\def\false{\mbox{\sc false}}

\def\EMPTY{\vskip8pt minus 4pt\relax}

\def\Er{\ensuremath{E^r}}
\def\Gr{\ensuremath{\mathcal{G}^r}}
\def\Ga{\mathcal{G}}

\def\N{\mathcal{N}}
\def\U{\mathcal{U}}
\def\A{\mathcal{A}}
\def\Comp{\mathcal{C}}

\def\Cpr{\Comp_p^r}

\def\stableSCC{\mbox{\tt inStableRoot}}
\def\distx_#1{\mbox{\tt{cd}}^{#1}} 
\makeatletter
\def\dist{\@ifnextchar_\distx{\distx_}}
\makeatother
\def\Knows#1{\mathsf{K}_p}

\gdef\dash---{\thinspace---\hskip.16667em\relax}

\theoremstyle{definition}
\newtheorem{definition}{Definition}
\theoremstyle{plain}
\newtheorem{theorem}{Theorem}
\newtheorem{lemma}[theorem]{Lemma}

\newtheorem{corollary}[theorem]{Corollary}
{\bfseries}{\itshape}
\newtheorem{assumption}{Assumption}{\bfseries}{\itshape}

\onlyLong{
\newconstruct{\ON}{\textbf{on}}{\textbf{do}}{\ENDON}{\textbf{end on}}

\newconstruct{\PROC}{\textbf{procedure}}{}{\ENDPROC}{\textbf{end on}}
\newconstruct{\PRED}{\textbf{predicate}}{}{\ENDPRED}{\textbf{end on}}
\newconstruct{\FUNC}{\textbf{function}}{}{\ENDFUNC}{\textbf{end on}}
}

\renewcommand{\leq}{\leqslant}
\renewcommand{\le}{\leqslant}
\renewcommand{\geq}{\geqslant}
\renewcommand{\ge}{\geqslant}

\gdef\dash---{\thinspace---\hskip.16667em\relax}

\gdef\op|{\,|\;}
\gdef\op:{\,:\;}
\newcommand{\li}[1]{\langle#1\rangle}

\def\la{\leftarrow}

\def\ra{\rightarrow}
\def\lt{\leadsto}

\def\set#1{\left\{#1\right\}}

\def\C{\mathcal{C}}
\def\G{\mathcal{G}}

\def\D{\mathcal{D}}
\def\S{\mathscr{S}}
\def\L{\mathcal{L}}
\def\P{\mathcal{P}}
\def\Q{\mathcal{Q}}

\newcommand{\M}{\mathcal{M}}
\newcommand{\R}{\mathcal{R}}

\newcommand{\msg}[1]{\langle#1\rangle}

\makeatletter
\def\leftNumbered{\tagsleft@true}\def\rightNumbered{\tagsleft@false}
\makeatother

\def\Timely{\mathcal{N}}

\def\Er{\ensuremath{E^r}}
\def\Gr{\ensuremath{\mathcal{G}^r}}
\def\Gcap(#1){\mathcal{G}^{\cap\, #1}}
\def\Ecap(#1){E^{\cap\, #1}}

\def\N{\mathcal{N}}
\def\Comp{\mathcal{C}}

\def\Cpr{\Comp_p^r}

\newcommand{\edge}[1]{\stackrel{#1}{\ra}}
\newcommand{\ltedge}[1]{\stackrel{#1}{\lt}}

\def\R{\mathcal{R}}

\newcommand{\goodD}{$D$-bounded}

\author{Martin Biely\thanks{EPFL, Switzerland. \texttt{martin.biely@epfl.ch}}
  \and Peter Robinson\thanks{Division of Mathematical Sciences, Nanyang Technological University, Singapore 637371. \texttt{peter.robinson@ntu.edu.sg}}
\and Ulrich Schmid\thanks{Technische Universit\"at Wien,
  Embedded Computing Systems Group (E182/2), Vienna, Austria. \texttt{s@ecs.tuwien.ac.at}}
}
\date{}

\title{Agreement in Directed Dynamic Networks\thanks{%
    A preliminary report of this paper has appeared at SIROCCO 2012.
    This work has been supported by the Austrian Science Foundation (FWF) 
    project P20529 and S11405.
 Peter Robinson has also been supported in part by Nanyang Technological 
 University grant M58110000 and Singapore Ministry of Education (MOE) 
 Academic Research Fund (AcRF) Tier 2 grant MOE2010-T2-2-082.}} %
\begin{document}
\maketitle
\begin{abstract}
  We study distributed computation in synchronous
  dynamic networks where an omniscient adversary controls the 
  unidirectional communication links. Its behavior is modeled as a sequence 
  of \emph{directed} graphs representing the active (i.e.\ timely) 
  communication links per round.
 We prove that consensus is impossible under some natural weak connectivity assumptions,
 and introduce vertex-stable root components as a 
means for circumventing this impossibility. 
Essentially, we assume that there is a short period of time during 
which an arbitrary part of the network remains strongly connected, while
its interconnect topology may keep changing continuously.
We present a consensus algorithm that works under this assumption, and prove
its correctness. Our algorithm maintains a local
estimate of the communication graphs, and applies %
techniques for detecting
stable network properties and univalent system configurations.  Our 
possibility
results are complemented by several impossibility results and lower bounds for
consensus and other distributed computing problems like leader election,
revealing that our algorithm is asymptotically optimal. 
%\keywords{distributed algorithm \and dynamic network \and consensus \and lower bound}
\end{abstract}

\section{Introduction}

Dynamic networks, instantiated, e.g., by (wired) peer-to-peer (P2P) 
networks, (wireless) sensor networks, mobile ad-hoc networks and vehicular 
area networks, are becoming
ubiquitous nowadays. The primary properties of such networks are (i) sets 
of participants
(called processes in the sequel) that are a priori unknown and maybe time-varying, and (ii) 
the absence of central control. Such assumptions make it very difficult to
setup and maintain the basic system, and create particular challenges
for the design of robust distributed services for applications running on 
such dynamic networks. 

A natural approach to build robust services despite the dynamic nature of
the system (e.g., mobility, churn, or failures of processes) is to use distributed consensus to agree system-wide on
(fundamental) parameters like schedules, frequencies, etc. Although system-wide agreement
indeed provides a very convenient abstraction for building robust services, 
it inevitably rests on the ability to efficiently implement consensus in a dynamic 
network.

Doing this in \emph{wireless} dynamic networks is particularly challenging, 
for several reasons: First, whereas wireline networks are usually 
adequately modeled by 
means of bidirectional links, this is not the case for wireless networks:
Fading phenomenons and interference \cite{GKFS10:perf} are local effects
that affect only the receiver of a wireless link. 
Given that the sender, or rather the receiver of the reverse link, is
to be found at a different location it is likely to
experience very different levels of fading and interference.
Thus, wireless links are more adequately modeled by means of pairs of
unidirectional links, which are considered independent of each other.

Second, wireless networks are inherently broadcast. When a process 
transmits, then every other process within its transmission range will
observe this transmission --- either by legitimately receiving the
message or as some form of interference. This creates quite irregular
communication behavior, such as capture effects and
near-far problems \cite{WJC00}, where certain (nearby) transmitters may
``lock'' some receiver and thus prohibit the reception of messages
from other senders. As a consequence, wireless
links that work correctly at a given time may have a 
very irregular spatial distribution, and may also vary heavily over 
time.

Finally, taking also into account mobility of processes and/or 
peculiarities in the system design (for example, duty-cycling is 
often used to conserve energy in wireless sensor networks), it is
obvious that static assumptions on the communication topology, as
underlying classic models like unit disc graphs, are not adequate for
wireless dynamic networks.

We hence argue that such dynamic systems can be modeled adequately
only by dynamically changing \emph{directed} communication
graphs.
Since synchronized clocks are
required for basic communication in wireless systems,
e.g., for transmission scheduling and sender/receiver synchronization,
we also assume that the system is synchronous. 

\subsection*{Main Contributions}
We consider a dynamic network modeled by a sequence of directed communication 
graphs, one for each round.

\begin{enumerate}
\item[(1)] We prove that communication graphs that are weakly connected in
every round are not sufficient for solving consensus, and introduce
a fairly weak additional assumption that allows to overcome this
impossibility. It requires that, in every round, there is exactly
one (possibly changing) strongly connected component (called a \emph{root component})
that has only out-going
links to (some of) the remaining processes and can reach every
process in the system via several hops. Since this assumption is still
too weak for solving consensus, we add the requirement that, eventually, there
will be a short interval of time where the processes in the
root component remain the same, although the connection topology
may still change. We coined the term \emph{vertex-stable root component}
(for some window of limited stability) for this requirement.

\item[(2)] We provide a consensus algorithm that works in this model, and
prove its correctness. Our algorithm requires a window of stability that
has a size of $4D$, where $D$ is the number of 
rounds required to reach all processes in the network from any process in the
vertex-stable root component.
\onlyLong{
While in general $D$ can be in as large as $n$, we show how to obtain an improved 
running time (in $O(\log n)$) when assuming certain expanding graph 
topologies.
}

\item[(3)] We show that any consensus and leader election algorithm has to
know an a priori bound on $D$. Since the system size $n$ is a
trivial bound on $D$, this implies that there is no uniform algorithm, i.e.,
no algorithm unaware of the size of the network, %
that solves consensus in our model.
In addition, we establish a lower bound of $D$ for the window of stability.

\item[(4)] We prove that neither reliable broadcast, atomic broadcast, nor 
  causal-order broadcast can be implemented in our model without 
  additional assumptions.  Moreover, there is no algorithm that solves 
  counting, $k$-verification, $k$-token dissemination, all-to-all token 
  dissemination, and $k$-committee election.
\end{enumerate}

\onlyShort{Lacking space did not allow us to include the detailed
algorithms and proofs in this paper; consult the full paper \cite{BRS2012:arxiv}
for all the details.}

\onlyLong{\section{Related Work}}
\onlyShort{\subsection*{Related Work}}

We are are not aware of any previous work on consensus in \emph{directed and 
  dynamic} networks with connectivity requirements as weak as ours.

\label{sec:relwork}

Dynamic networks have been studied intensively in distributed computing. Early
work on this topic includes~\cite{afek+gr:slide,awerbuch+pps:dynamic}. 
\onlyLong{One basic assumption that can be used to categorize research in
dynamic networks is whether the set of processes is assumed to be
fixed, or subject to churn (i.e., processes join and leave over time).
The latter has mostly been considered in the area of peer-to-peer
networks and the construction of overlays. We refer the interested
reader to \cite{KSW10} for a more detailed treatment of related work
in this area. 

When the set of processes is considered to be fixed, dynamicity in the
  network is modeled by changes in the network topology.
Over time several approaches to modeling dynamic connectivity in
networks have been proposed.} 
We will in the following focus on two lines of research that are closest to ours:
work that models directly the underlying (evolving)
communication graph, and approaches taken in the context of consensus. 

\onlyLong{\subsection*{Evolving graph models}}
There is a rich body of literature on dynamic graph models going
     back to~\cite{HG97:Dyn}, which also mentions for the first time
     modeling a dynamic graph as a sequence of static graphs, as we do.
A survey on dynamic networks can be found in \cite{kuhn-survey}.
Recently, Casteigts et al.~\cite{CFQS11:TVG} have introduced a
     classification of the assumptions about the temporal properties of time
     varying graphs.
\onlyShort{In the full paper \cite{BRS2012:arxiv}, we}\onlyLong{We will 
(cf.\ Lemma~\ref{lem:classification})} show that our assumption falls 
between two of
the weakest classes considered, as we can only guarantee
one-directional reachability.\footnote{Here, reachability does not refer to the
  graph-theoretic concept of reachability (in a single graph), but rather to the ability to
  eventually communicate information to another process (which is a
  property of the sequence of communication graphs.}
We are not aware of other work considering such weak assumptions
in the context of agreement.

Closest to our own work is that of Kuhn et al.~\cite{KOM11}, who also
     consider agreement problems in dynamic networks based on the 
model of~\cite{kuhn+lo:dynamic}. 
This model is based on distributed computations organized in lock-step
rounds, and states assumptions on the connectivity in each round
as a separate (round) communication graph.
While the focus of \cite{kuhn+lo:dynamic} is the complexity of
aggregation problems in dynamics networks, \cite{KOM11} focuses on
agreement problems; more specifically 
on the $\Delta$-coordinated consensus problem, which extends consensus by requiring all processes to decide within
$\Delta$ rounds of the first decision. 
In both papers, only undirected graphs that are \emph{connected in every round}
are considered. In terms of the classes of~\cite{CFQS11:TVG}, the model of
\cite{kuhn+lo:dynamic} is in one of the strongest classes (Class~10)
in which every process is always reachable by every other process. 
Since node failures are not considered, solving consensus is
always %
possible in this model without additional assumptions, thus the focus
of~\cite{KOM11} is on $\Delta$-coordinated consensus problem (all processes are required to
decide within $\Delta$ rounds) and its time complexity.
In sharp contrast to this work, our communication graphs are 
\emph{directed}, and our rather weak connectivity assumptions do not
guarantee bidirectional (multi-hop) communication between all
processes.

\onlyLong{Dynamic networks have been studied in \cite{p2p-soda}
in the context of ``stable almost-everywhere agreement'' (a variant of the ``almost
everywhere agreement'' problem introduced by \cite{DPPU88}), which weakens the
classic consensus problem in the sense that a small linear fraction of processes
might remain undecided. In the model of \cite{p2p-soda},
the adversary can subject up to a linear fraction of nodes to churn per round;
assuming that the network size remains stable, this means that up to
$\varepsilon n$ nodes (for some small $\varepsilon>0$) can be replaced by new
nodes in every round.  Moreover, changes to the undirected graph of the
network are also under the control of the adversary. To avoid almost-everywhere
agreement from becoming trivially unsolvable, \cite{p2p-soda} assumes that the
network is always an expander.

The leader election problem in dynamic networks has been studied in
\cite{CRW11:FOMC} where the adversary controls the mobility of nodes in a
wireless ad-hoc network. This induces dynamic changes to the (undirected)
network graph in every round and requires any leader election algorithm to
take $\Omega(D n)$ rounds in the worst case, where $D$ is an upper bound on
information propagation.
}

\onlyLong{\subsection*{Transmission Failure Models}}
Instead of considering a dynamic graph that defines which
     processes communicate in each round, an alternative approach
      is based on the (dual) idea of assuming a fully connected
     network of (potential) communication, and considering that
     communication/message transmission in a round can fail. 
The notion of transmission failures was introduced by Santoro and
     Widmayer~\cite{SW89}, who assumed dynamic transmission failures
     and showed that $n-1$ dynamic transmission
     failures in the benign case (or $n/2$ in case of 
     arbitrary transmission failures) render any non-trivial agreement 
     impossible. 
As it assumes unrestricted transmission failures (the 
(only) case considered in their proof are failures that
affect all the transmissions of a \emph{single} process), 
it does not apply to any model which considers
perpetual mutual reachability of processes (e.g.,~\cite{KOM11}).

The HO-model~\cite{CBS09} is also based on transmission failures. It relies 
on the collection of sets of processes a process \emph{hears
     of} (i.e., receives a message from) in a round.
Different system assumptions are modeled by predicates over this
     collection of sets. The HO-model is totally
oblivious to the actual reason \emph{why} some process does not 
hear from another one: Whether the sender committed a send omission or crashed, 
the message was
     lost during transmission or is simply late, or the receiver committed a
receive omission. 
A version of the model also allowing communication to be corrupted is presented
in~\cite{BCGHSW07:PODC}. Indeed, the HO-model is very close to our
graph model, as an edge from $p$ to $q$ in the graph of round $r$
corresponds to $p$ being in the round $r$ HO set of $q$. 
 
The approach taken by Gafni~\cite{Gaf98} has some similarities to the
     HO-model (of which it is a predecessor), but is more focused on
     process failures than the two approaches above.
Here an oracle (a round-by-round failure detector) is considered to tell processes the set
     of processes they will be not be able to receive data from in the
     current round. 
Unlike the approaches discussed above, it explicitly states how
     rounds are implemented; nevertheless, the oracle abstracts
     away the actually reason for not receiving a message.  So, like in 
     the HO-model, the same device is used to describe failures and 
     (a)synchrony. 

Another related model is the perception based
failure model~\cite{SWK09,BSW11:hyb}, which uses a sequence of
perception matrices (corresponding to HO sets) to express failures
of processes and links. As for transmission failures, the
impossibility result of Santoro and Widmayer is circumvented
by putting separate restrictions on the number of outgoing and
     incoming links that can be affected by transmission
     failures~\cite{SWK09}. 
Since transmission failures are counted on a per process/per round
     basis, agreement was shown to be possible in the presence of $O(n^2)$
     total transmission failures per round. 

The approach we used in \cite{BRS11:IPDPS} relied on restricting
the communication failures per round in a way that secures the existence
of a \emph{static} skeleton graph (which exists in all rounds). 

\section{Model and Preliminaries}
\label{sec:model}

We consider synchronous computations of a dynamic network of a fixed set of
      distributed processes $\Pi$ with $|\Pi|=n\geq 2$. 
Processes can communicate with their current neighbors in the network
      by sending messages taken from some finite message alphabet
      $\M$. 

\onlyLong{In the following three subsections, we will present our 
  computational
model and define what it means to solve consensus in this model.
In \Section~\ref{sec:connprops}, we will introduce
constraints that make consensus solvable in our model.
}

\onlyLong{\subsection{Computational Model.}}

Similar to the $\mathcal{LOCAL}\/$ model \cite{Pel00}, we assume that processes organize their computation as an infinite
sequence of lock-step rounds. For every $p\in\Pi$ and each round
$r>0$, let $S_p^r \in \S_p$ be the state of $p$ at the beginning of
round $r$ taken from the set of all states $p$ can possibly enter; the
initial state is denoted by $S_p^1\in \S_p^1\subset\S_p$, with
$\S_p^1$ being the set of potential initial states of $p$. The round $r$ 
computation of process $p$ is determined by the following two functions 
that make up $p$'s algorithm: The message sending function $M_p:\S_p \to 
\M$ determines the message $m_p^r$ broadcast by $p$ in round $r$, based on 
$p$'s state $S_p^{r}$ at the beginning of round $r$.  We assume that some 
(possibly empty) message is broadcast in every round, to all
(current!) out-neighbors of $p$.
The neighborhood of a process in round $r$ depends solely on the underlying
communication graph of round $r$\onlyShort{.}\onlyLong{, which we define 
in \Section~\ref{sec:comm-concepts}.} Note that the only part of the
round $r$ communication graph $p$ can (directly) observe is the set of
its in-neighbors.
The transition function $T_p:\S_p\times 2^{(\Pi\times\M)} \to \S_p$ takes $p$'s state 
$S_p^{r}$ at the beginning of round $r$ and a set of pairs of process
ids and messages $\mu_p^r$. This set represents the round $r$ 
messages\footnote{We only consider messages sent in round $r$ here, so we assume
      communication-closed~\cite{EF82} rounds.}
received by $p$ from other processes in the system, 
and computes the successor state $S_p^{r+1}$. 
We assume that, for each process $q$, there is at most one
$(q,m_q^r)\in\mu_p^r$ such that $m_q^r$ is the message $q$ sent in round $r$.
Note that neither $M_p$ nor $T_p$ need to involve $n$, i.e., the algorithms
executed by the processes may be uniform with respect to the network size $n$.

\onlyLong{\subsection{Consensus}}

To formally introduce the consensus problem, we assume
some ordered set $V$ and consider the set
     of possible initial states $\S_p^1$ (of process $p$) to be
     partitioned into $|V|$ subsets $\S_p^1[v]$, with $v\in V$.
When $p$ starts in a state in $\S_p^1[v]$, we say that $v$ is $p$'s
input value, denoted $v_p=v$. Moreover, we assume that, for each $v\in
V$, there is a set $\D_p[v]\subset\S_p$ of decided states such that
$\D_p[v]\cap\D_p[w]=\emptyset$ if $v\ne w$ and $\D_p[w]$ is closed under
$p$'s transition function, i.e., $T_p$ maps any state in this subset
to this subset (for all sets of messages). We
say that $p$ has \emph{decided} on $v$ when it is in some state in $\D_p[v]$. When
$p$ performs a transition from a state outside of the set of decided
states to the set of decided states, we say that $p$ \emph{decides}.
We say that an \emph{algorithm
$\mathcal{A}$ solves consensus} if the following properties hold in every run of
$\mathcal{A}$:
\begin{description}%
\item[Agreement:] If process $p$ decides on $x_p$ and $q$ decides on
  $x_q$, then $x_p=x_q$.
\item[Validity:] If a process decides on $v$, then $v$ was proposed
    by some $q$, i.e., $v_q=v$.
\item[Termination:] Every process must eventually decide. 
\end{description}

At a first glance, solving consensus might appear easier in our 
model than in the classic crash failure model, where processes simply stop executing
the algorithm. This is not the case, however, as we can model
(similarly to~\cite{CBS09}) crash failures in the following way: when
process $q$ crashes in round $r$ and takes no more steps no other process
ever receives messages from $q$ after $r$, which is
equivalent to considering that $q$ does not have any outgoing edges from round $r+1$
on. While $q$ itself can still receive messages and perform computations, the
remaining processes are not influenced by $q$ from round $r$
on.

\onlyLong{\subsection{Communication Model}
\label{sec:comm-concepts}
}
\onlyShort{ \subsubsection*{Communication Model.}}

\tikzstyle{p}=[circle,draw=gray,fill=lightgray!30,thick,inner sep=0pt,minimum size=4mm]
\tikzstyle{p6}=[p,double]
\tikzstyle{link}=[->,black,thick,auto]
\tikzstyle{fail}=[->,black,thick,densely dotted,auto]
\tikzstyle{i}=[draw=none,opacity=0]
\newcommand{\showproc}[5]{%
  \node[#1]  (p1)               {$p_1$};
  \node[#3]  (p3) [right=of p1] {$p_3$};
  \node[#5]  (p5) [right=of p3] {$p_5$};
  \node[#4]  (p4) [below=of p3] {$p_4$};
  \node[#2]  (p2) [above=of p3] {$p_2$};

%  \draw[red,use as bounding box] (p2.north east) -- +(2ex,4ex) -- (p4.south west) -- +(-2ex,-7ex);
   
  \begin{pgfonlayer}{background}
  \draw[draw=none,fill=lightgray!20] (current bounding box.south west) rectangle (current bounding box.north east);
  \end{pgfonlayer}
}  
\setlength{\textfloatsep}{\baselineskip}
\begin{figure*}[t]
  \tiny
\centering{
%\begin{tabular*}{\textwidth}{p{.25\textwidth}|@{\hspace{.05\textwidth}}p{.75\textwidth}}
%  \begin{minipage}[t]{0.20\textwidth}
\subfloat[$\Ga^1$]{%
\begin{tikzpicture}
  \showproc{p}{p}{p}{p}{p}

  \draw[link] (p1) to [bend left] (p2);
  \draw[link] (p2) to [bend left] (p1);
  \draw[link] (p4) to (p1);
  \draw[link] (p4) to (p5);
  \draw[link] (p2) to (p3);
  \draw[link] (p5) to (p2);
\end{tikzpicture}
\label{fig:g1}
}\ 
\subfloat[$\Ga^2$]{%
\begin{tikzpicture}
  \showproc{p}{p}{p}{p}{p}

  \draw[link] (p1) to (p2);
  \draw[link] (p2) to (p3);
  \draw[link] (p4) to (p1);
  \draw[link] (p4) to (p5);
\end{tikzpicture}
\label{fig:g2}
}\ 
\subfloat[$\Ga^3$]{%
\begin{tikzpicture}
  \showproc{p}{p}{p}{p}{p}

  \draw[link] (p2) to (p1);
  \draw[link] (p3) to (p1);
  \draw[link] (p5) to (p3);
  \draw[link] (p5) to (p2);
  \draw[link] (p3) to (p4);
\end{tikzpicture}
\label{fig:g3}
}
%\end{minipage}
%\end{tabular*}
\caption{
\label{fig:graphs}%
}
} % centering
\end{figure*}

The evolving nature of the network topology is modeled as an infinite
      sequence of simple directed graphs $\G^1,\G^2,\dots$, which is fixed 
      by an adversary having access to the processes' states.
For each round $r$, we denote the \emph{round $r$ communication graph} by
      $\Gr=\li{V,\Er}$, where each node of the set $V$ is associated
      with one process from the set of processes $\Pi$, and where
      $\Er$ is the set of directed edges for round $r$, such that
there is an edge from $p$ to $q$, denoted as $(p\ra q)$,
      \textit{iff} $q$ receives $p$'s round $r$ message (in round $r$).
\onlyLong{
Figure~\ref{fig:graphs} shows a sequence of graphs for a network of $5$
processes, for rounds $1$ to $3$.
}
For any (sub)graph $G$, we will use the notation $V(G)$ and $E(G)$ to
refer to the set of vertices respectively edges of $G$, i.e., it
always holds that $G=\li{V(G),E(G)}$.

\onlyLong{
Note that, 
for deterministic algorithms, a run is completely determined by the input 
values assigned to the processes and the sequence of communication graphs.
}

To simplify the presentation, we will denote a process
      and the associated node in the communication graph by the same
      symbols and omit the set from which it is taken if there is no
      ambiguity.  
We will henceforth write $p \in \Gr$ and $(p\ra q) \in \Gr$ 
      instead of $p \in V(\Gr)$
      resp.\ $(p\ra q) \in \Er$.

The \emph{neighborhood of $p$ in round $r$} is the set of processes $\Timely_p^r$ that $p$
receives messages from in round $r$, formally,
$\Timely_p^r = \set{q \mid (q\ra p) \in \Ga^r }$.

Similarly to the classic notion of ``happened-before'' introduced in \cite{Lam78}, we 
say
that a process $p$ \emph{(causally) influences process $q$ in round $r$}, 
expressed by $(p\ltedge{r}q)$ or just $(p\lt q)$ if $r$ is clear from the context, iff  either
(i) $p \in \Timely_q^r$, or (ii) if $q=p$. %
We say that there is a \emph{(causal) chain of length $k\geq 1$ starting from $p$ in round
$r$ to $q$}, denoted by $(p\ltedge{r[k]}q)$.
if there exists a sequence of not necessarily distinct processes
$p=p_0,\dots,p_{k}=q$ such that $p_i\ltedge{r+i}p_{i+1}$, for
all $0\le i <k$. 
Conversely, we write $(p\not\ltedge{r}q)$ or simply $(p\not\lt q)$
when there is no such $k$.

The \emph{causal distance} $\dist_r(p,q)$ at round $r$ from
process $p$ to process $q$ is the length of the shortest causal
chain starting in $p$ in round $r$ and ending in 
$q$, formally, 
$$\dist_r(p,q) := \min(\{k \mid (p\ltedge{r[k]}q)\} \cup \{\infty\}).$$
Note that we have $\dist_r(p,p)=1$.
The following Lemma~\ref{lem:cd} shows that the causal distance in successive rounds
cannot arbitrarily decrease.

\begin{lemma}[Causal distance in successive rounds]\label{lem:cd}
For every round $r\geq 1$ and every two processes $p,q\in\Pi$, 
it holds that $\dist_{r+1}(p,q) \geq \dist_r(p,q)-1$. As a
consequence, if $\dist_r(p,q)=\infty$, then also $\dist_{r+1}(p,q)=\infty$.
\end{lemma}
\onlyLong{
\begin{proof}
Since $(p\lt p)$ in every round $r$, the definition of 
causal distance trivially implies $\dist_r(p,q) \leq 1 + \dist_{r+1}(p,q)$.\qed
\end{proof}
}

Note that, in contrast to the similar notion of dynamic distance defined in
     \cite{kuhn-survey}, the causal distance in \emph{directed} graphs
     is not necessarily symmetric. Moreover, 
if the adversary chooses the graphs $\G^r$ such that not all
     processes are strongly connected, the causal distance between two
     processes can even be finite in one and infinite in the other direction.
In fact, even if $\G^r$ is strongly connected for round $r$ (but not
     for rounds $r'>r$), $\dist_r(p,q)$ can be infinite.
We will not consider the whole communication graph to be
strongly connected in this paper, we make use of the notation of \emph{strongly connected
components (SCC)}. 
We write $\Cpr$ to denote the unique SCC of $\Ga^r$
that contains process $p$ in round $r$ or simply $\C^r$ if $p$ is irrelevant.

It is apparent that $\dist_r(p,q)$ and $\dist_r(q,p)$ may be infinite
even if $q\in\Cpr$.
In order to be able to argue (meaningfully) about the maximal length of causal
chains within an SCC, we thus introduce a ``continuity property'' over rounds.
This leads us to the crucial concept of an
\emph{$I$-vertex-stable strongly connected component}, denoted 
as $\C^I$: an SCC $\C^I$ is vertex-stable during $I$ requires that 
$\forall p\in\C^I, \forall r\in I: V(\Cpr)=V(\C^I)$.
That is it requires that the set of vertices of a strongly connected 
component $\C^I$ remains stable throughout all rounds in the nonempty interval $I$.  
Note that its topology may undergo changes, but must form an SCC in every 
round.  Formally, $\C^I$ being vertex-stable during $I$ requires that 
$\forall
p\in\C^I, \forall r\in I: V(\Cpr)=V(\C^I)$.
The important property of $\C^I$ is that information is guaranteed to
spread to all vertices of $\C^I$ if the interval $I$ is large
enough (cf.\ Lemma~\ref{lem:infprop} below).

Let the \emph{round $r$ causal diameter} $D^r(\C^I)$ of a
vertex-stable SCC $\C^I$ be the largest causal distance $\dist_r(p,q)$
for any $p,q\in\C^I$. The \emph{causal diameter} $D(\C^I)$ of
a vertex-stable SCC $\C^I$ in $I$ is the largest causal distance 
$\dist_x(p,q)$
starting at any round $x\in I$ that ``ends'' in $I$, i.e., 
$x+\dist_x(p,q)-1\in I$.
If there is no such causal distance (because $I$ is too short), $D(\C^I)$ is 
assumed to be infinite. Formally, for $I=[r,s]$ with $s\geq r$,\footnote{Since
$I$ ranges from the beginning of $r$ to the end of $s$, we define
$|I|=s-r+1$.}
\begin{align*}
  D(\C^I) = \min\{\max\{D^x(\C^I) \mid\ &\text{$x\in[r,s]$ and} \\
  &\text{$x+D^x(\C^I)-1\leq s$}\} \cup \{\infty\} \}.
\end{align*}

If $\C^I$ consists only of one process, then we obviously have 
$D(\C^I)=1$.
The following Lemma~\ref{lem:boundDI} establishes a bound for $D(\C^I)$
also for the general case.

\begin{lemma}[Bound on causal diameter]
  \label{lem:boundDI}
Given some $I=[r,s]$ and $\C^I$ a vertex-stable SCC with $|\C^I|\geq 2$:
If $s\geq r+|\C^I|-2$, then $D(\C^I)\leq|\C^I|-1$. 
\end{lemma}
\onlyLong{
\begin{proof}
Fix some process $p\in\C^I$ and some $r'$ where $r\leq r' \leq 
s-|\C^I|+2$.  Let $\P_0=\set{p}$, and define for each $i>0$ the set
$\P_i=\P_{i-1} \cup \{q: \exists q'\in\P_{i-1}: q'\in\N_{q}^{r'+i-1}
\cap \C^I\}$. 
$\P_i$ is hence the set of processes $q\in\C^I$ such that $(p\ltedge{r'[i]}q)$ holds. 
Using induction, we will show that  $|\P_k| \geq \min\{|\C^I|,k+1\}$
for $k\geq0$. Induction base $k=0$: $|\P_0| \geq \min\{|\C^I|,1\}=1$ follows immediately from
$\P_0=\set{p}$.
Induction step $k \to k+1$, $k\geq 0$:  
Clearly the result holds if $|\P_k|=|\C^I|$, thus we
consider round $r'+k$ and $|\P_k|<|\C^I|$: 
It
follows from strong connectivity of $\Ga^{r'+k} \cap \C^I$ that there is a set of
edges from processes in $\P_k$ to some non-empty set
$\L_k\subseteq\C^I\setminus\P_k$. Hence, we have
$\P_{k+1}=\P_k\cup\L_k$, which implies $|\P_{k+1}| \geq |\P_k| + 1 \geq k+1 + 1 = k+2 = \min\{|\C^I|,k+2\}$
by the induction hypothesis.

Thus, in order to guarantee $\C^I=\P_{k}$ and thus $|\C^I|=|\P_{k}|$, choosing
$k$ such that $|\C^I|=1+k$ and $k\le s-r'+1$ is sufficient. 
Since $s\ge r'+|\C^I|-2$, both conditions can be fulfilled by choosing $k=|\C^I|-1$.
Moreover, due to the definition of $\P_k$, it follows that $\dist_{r'}(p,q)\le|\C^I|-1$
for all $q\in\C^I$. Since this holds for any $p$ and any $r'\leq s-|\C^I|+2$, 
we finally obtain $|\C^I|-1\ge D^{r'}(\C^I)$ and hence $|\C^I|-1\ge
D(\C^I)$, which completes the proof of Lemma~\ref{lem:boundDI}.

\end{proof}
}

Given this result, it is tempting to assume that, for a
vertex-stable SCC $\C^I$ with finite causal diameter $D(\C^I)$, any information propagation
that starts at least $D(\C^I)-1$ rounds before the final round of $I$ will reach
all processes in $\C^I$ within $I$. This is not generally true, however, as
the following example for $I=[1,3]$ and a vertex-stable SCC of four processes shows: If
$\Ga^1$ is the complete graph whereas $\Ga^2=\Ga^3$ is a ring, $D(\C^I)=1$,
but information propagation starting at round 2 does not finish by the end
of round $3$. However, the following Lemma~\ref{lem:infprop} gives a bound on
the latest starting round that guarantees this property.

\begin{lemma}[Information propagation]\label{lem:infprop}
Suppose that $\C^I$ is an $I$-vertex-stable strongly connected component 
of size $\ge 2$ that has $D(\C^I)<\infty$, for $I=[r,s]$, and let $x$ be 
the maximal round where $x+D^x(\C^I)-1 \leq s$.  Then, 
\begin{compactitem}
\item[(i)] for every $x'\in [r,x]$, it holds that $x'+D^{x'}(\C^I)-1\leq 
  s$ and $D^{x'}(\C^I)\leq D(\C^I)$ as well, and
\item[(ii)] $x\geq \max\{s-|\C^I|+2,r\}$.
\end{compactitem}
\end{lemma}
\onlyLong{
\begin{proof}
Since $D(\C_I)<\infty$, the maximal round $x$ always exists.  
Lemma~\ref{lem:cd} reveals that for all $p,q \in \C^I$, we have
$x-1+\dist_{x-1}(p,q)-1 \leq x+\dist_x(p,q)-1\leq s$, which implies 
$x'+\dist_{x'}(p,q)-1\leq s$ for every $x'$ where $r\leq x' \leq x$ and 
proves (i).  The bound given in (ii)
follows immediately from Lemma~\ref{lem:boundDI}.
\end{proof}
}

Since we will frequently require a vertex-stable SCC $\C^I$ that guarantees 
bounded information propagation also for late starting rounds, we introduce the
following Definition~\ref{def:goodD}.

\begin{definition}[\goodD\ $I$-vertex-stable SCC]
  \label{def:goodD}
An $I$-vertex-stable SCC $\C^I$ with $I=[r,s]$ 
  is \emph{\goodD} if $D\geq D^I(\C^I)$ and $D^{s-D+1}(C^I)\leq D$.
\end{definition}

\section{Required Connectivity Properties}
\label{sec:connprops}

Up to now, we did not provide any guarantees on the connectivity of the 
network, the lack of which makes consensus
trivially impossible.\onlyLong{\footnote{The adversary can simply choose 
    the empty set for
the set of edges in every round.}} In this section, we will add some weak
constraints on the adversary that circumvent this impossibility. Obviously,
we want to avoid requesting too strong properties of the network topology (such 
as stating that $\G^r$ is strongly connected in every round $r$), as
this would not only make consensus trivially solvable but as we have argued
would reduce the applicability of our results in wireless networks.

As a first attempt, we could assume that, in every round $r$, the communication
graph $\G^r$ is weakly connected. This, however, turns out to be insufficient. Even
if the adversary choses a \emph{static} topology, it is easy to see that
consensus remains impossible: Consider for example the graph that is partitioned
into $3$ strongly connected components $\C_0$, $\C_1$, and $\C_2$ such that
there are only outgoing edges from $\C_0$ respectively $\C_1$ pointing to
$\C_2$, whereas $\C_2$ has no outgoing edges. If all processes in $\C_0$ start
with $0$ and all processes in $\C_1$ start with $1$, this  yields a
contradiction to agreement: For $i \in \set{0,1}$, processes in $\C_i$ can never
learn the value $1-i$, thus, by an easy indistinguishability argument, it
follows that processes in $\C_0$ and $\C_1$ must decide on conflicting values.

In order to define constraints that rule out the existence of $\C_0$ and 
$\C_1$ as above, the concept of \emph{root components} proves useful:
Let $\R^r \subseteq \G^r$ be an SCC that has no
  incoming edges from any $q \in \G^r\setminus \R^r$. We say that $\R^r$ is a
  \emph{root component} in round $r$, formally:
  $$\forall p \in \R^r\ \forall q \in \Ga^r\colon (q~\ra~p) \in \Ga^r 
  \Rightarrow q \in
  \R^r.$$
\onlyLong{
For example, in Figure~\ref{fig:g1}, process $p_4$ forms a root component by
itself, while processes $p_1$ and $p_2$ form a SCC that is not a root
component since it has incoming edges.
}

The following Lemma~\ref{lem:root} establishes some simple facts on $\Gr$.

\begin{lemma}\label{lem:root}
Any $\Gr$ contains at least one and at most $n$ root components (isolated processes), 
which are all disjoint.
If $\Gr$ contains a single root component $\R^r$, then 
$\Gr$ is weakly connected, and there is in fact a directed 
(out-going) path from every $p\in \R^r$ to every $q \in \Gr$.
\end{lemma}

\begin{proof}
We first show that every weakly connected directed simple graph $G$ has at least
one root component. To see this, contract every SCC to a 
single vertex and remove all resulting self-loops. The resulting graph $G'$ is a directed acyclic graph (DAG)
(and of course still weakly connected), and hence $G'$ has at least one vertex
$C$ (corresponding to some SCC in $G$) that has no incoming edges. By
construction, any such vertex $C$ corresponds to a root component in the
original graph $G$. Since $\Gr$ has at least $1$ and at most $n$ weakly
connected components, the first statement of our lemma follows.

To prove the second statement, we use the observation that there is a directed
path from $u$ to $v$ in $G$ if and only if $C_w = C_v$ there is a directed path from $C_u$
to $C_v$ in the contracted graph $G'$.
If there is only $1$ root component in $G$, then the above discussion implies
that there is exactly one vertex $\R$ in the contracted graph $G'$
that has no incoming edges.
Since $G'$ is connected, $\R$ has a directed path to every other vertex in $G'$,
which implies that every process $p \in \R$ has a directed path to every vertex
$q$, as required.
\end{proof}

Returning to the consensus impossibility example for weakly connected graphs above,
it is apparent that the two components $\C_0$ and
$\C_1$ are indeed both root components. Since consensus is not
solvable in this case, we assume in the sequel
that there is at most \emph{a single} root component in $\G^r$, for any 
round $r$. We already know (cf.\ \cite{BRS11:IPDPS}) that this assumption 
makes consensus solvable if the topology (and hence the root component) is 
static.
In the terminology of this paper, \cite{BRS11:IPDPS} stipulates a
stronger version of the existence of an $\infty$-interval
vertex-stable root component, as some links were required to remain
unchanged throughout the computation.

Since we are interested in \emph{dynamic} networks, however, we assume in this paper
that the
root component may change throughout the run, i.e., the (single)
root component $\R^r$ of $\G^r$ might consist of a different set of processes in
every round round $r$.
\onlyLong{Figure~\ref{fig:graphs} shows a sequence of graphs where
there is exactly one root component in every round. }
It is less straightforward to reason about the
solvability of consensus in this case. We will establish in
\Section~\ref{sec:imposs} that consensus is again 
impossible to solve without further constraints.

As root components are special cases of strongly connected components,
     we define an \emph{$I$-vertex-stable root component $\R^I$} as an
     $I$-vertex-stable strongly connected component that is a
     root component in every round $r\in I$. 
Clearly, all the definitions and results for vertex-stable components carry over to
     vertex-stable root components. 

Restricting our attention to the case where exactly one vertex-stable
root component $\R^I$ exists, it immediately follows from 
Lemma~\ref{lem:root} that information
of any process in $\R^I$ propagates to all processes if $I$ is
large enough. More specifically, we can extend our notions of causal 
diameter of a vertex-stable
SCC to the whole network: The \emph{round $r$ network causal diameter 
$D^r$} is the largest $\dist_r(p,q)$ for any $p\in\R^r$ and any $q\in\Pi$.  
Similarly to the causal diameter of a vertex-stable component of an
interval, we define the \emph{network causal diameter $D^I$} for an 
interval $I$ as
the {largest round $x$, $x\in I$, network causal diameter
that also (temporally) falls within $I$}, i.e., satisfies 
$x+D^x-1\in I$ and hence $x+\dist_x(p,q)-1\in I$
for any $p\in\R^r$ and any $q\in\Pi$.
\onlyShort{
It is straightforward to establish versions of Lemma~\ref{lem:boundDI}  
and~\ref{lem:infprop} for root components and their causal influence.
}

\onlyLong{
The following versions of Lemma~\ref{lem:boundDI}  and~\ref{lem:infprop} 
for root components and their causal influence
on the whole network can be established analogously to the results for SCCs:

\begin{lemma}[Bound on network causal diameter]\label{lem:boundDIN}
Assume that there is a single vertex-stable root component $\R^I$ for some 
$I=[r,s]$.
If $s\geq r+n-2$, then $D^I\leq n-1$. 
\end{lemma}
Note that Lemma~\ref{lem:boundDIN} considers the worst case where the 
network topologies can even correspond to a line.  This assumption might 
be too pessimistic for many real world networks. By assuming graph 
topologies that allow fast information spreading, we can get much
better causal network diameters like $D^I \in O(\log n)$ 
(cf.\ \Section~\ref{sec:improved}).  

\begin{lemma}[Network information propagation]\label{lem:infpropN}
Assume that there is a single vertex-stable root component $\R^I$
in $I=[r,s]$ with network causal diameter $D^I<\infty$.
Let $x$ be the maximal round where $x+D^x-1\leq s$. Then,
\begin{compactitem}
\item[(i)] for every $x'\in [r,x]$, it also holds that also $x'+D^{x'}\leq 
s$ and $D^{x'}\leq D^I$, and 
\item[(ii)] $x\geq \max\{s-n+2,r\}$.
\end{compactitem}
\end{lemma}

As in the case of $I$-vertex-stable SCCs we also define the
\goodD\ variant of root components, which are central to our model. 

\begin{definition}[\goodD\ $I$-vertex-stable root 
  components]%
A vertex-stable root component $\R^I$ in $I=[r,s]$, $s\geq r$, 
is \goodD\ if $D\geq D^I$ and $D^{s-D+1}\leq D$.
\end{definition}
}
Note that a plain $I$-vertex-stable root component with $I\geq n-1$
is always \goodD\ for $D=n-1$\onlyLong{, recall Lemma~\ref{lem:infpropN}}.  
The purpose of the 
definition is to also allow smaller choices of $D$.

We will show in \Section~\ref{sec:imposs} that the following 
Assumption~\ref{ass:window} is indeed very weak, in the sense that many 
problems considered in distributed computing remain unsolvable.

\begin{assumption} \label{ass:window}
For any round $r$, there is exactly one root component $\R^r$ in
$\G^r$, and all vertex-stable root components $\R^I$ with $|I|\geq D$ are
\goodD.
Moreover, there exists an interval of rounds $J=[r_{ST},r_{ST}+d]$, with $d>4D$,
such that there is a \goodD\ $J$-vertex-stable root component.
\end{assumption}

\onlyLong{\section{Reaching Consensus} 
  \label{sec:overview}

The underlying idea of our consensus algorithm is to use flooding to
     propagate the largest input value throughout the network. 
However, as Assumption~\ref{ass:window} does not guarantee
     bidirectional communication between every pair of processes,
     flooding is not sufficient: 
The largest proposal value could be hidden at a single process $p$
     that never has outgoing edges. 
If such $p$ never accepts smaller values, clearly we cannot reach
     agreement (without potentially violating validity). 
Thus we have to find a way to force $p$ to accept also a smaller
     value. 
One well-known technique to do so is locking a value. 
Again we do not want $p$ to lock its value, but some process(es) that
     we know will be able to impose the locked value. 
That is processes that can successfully flood the system with the
     locked value. 
As we have seen in \ref{lem:infpropN} processes in $I$-vertex stable
     root components can do this, if $I$ is long enough. 
As we will see, Assumption~\ref{ass:window} does indeed guarantee the
     existence of such interval. 
Note that the processes that lock a value can only decide on their
     locked value once they are sure that every other process has
     accepted this value as well. 
It is apparent that---in order to lock a value and to later decide a
     value---processes need to know that they are in a vertex stable
     root component. 
Our first step (cf.\ Section~\ref{sec:approxalgo}) is thus to present an
     algorithm that allows processes to detect just that. 
Then (cf.\ Section~\ref{sec:consalgo}) we provide a second algorithm that based
     on information about stable root components solves consensus. 
As we will see the main complication in this approach comes from the
     fact that a process can only detect whether it is part of the
     root component of round $r$ reliably once $r$ has already passed.
}

\onlyLong{\subsection{The Local Network Approximation Algorithm}
\label{sec:approxalgo}
}
\onlyShort{\section{Solving Consensus by Network Approximation}
  \label{sec:thealgo}
} 

\onlyLong{As we have explained above the main goal of this section is
     to provide information about root components of (previous)
     rounds. 
Initially, every process $p$ has no knowledge of the network. 
In order to provide information about root components, process $p$
     needs to \emph{locally} acquire knowledge about the information
     propagation in the network. 
As we have seen, $p$'s information is only  guaranteed to propagate
     throughout the network if $p$ is in a $I$-vertex stable root
     component with finite network causal diameter $D^I$. 
Thus, for $p$ to locally acquire knowledge about information
     propagation, it has to acquire knowledge about the (dynamic)
     communication graph. 
}  

\onlyShort{Initially, every process $p$ has no knowledge of the network --- it
only knows its own input value. Any algorithm that
correctly solves consensus must guarantee that, when $p$ makes its
decision, it either knows that its value has been/will be adopted by all other
processes or it has agreed to take over some other process' decision value. \onlyLong{In
either case, process $p$ needs to \emph{locally} acquire knowledge about the information
propagation in the network.} As we have seen, $p$'s information is only 
guaranteed to propagate throughout the network if $p$ is in a $I$-vertex
stable root component with finite network causal diameter $D^I$. 
Thus, for $p$ to locally acquire knowledge about information
propagation, it has to acquire knowledge about the (dynamic)
communication graph. 
}

We allow $p$ to achieve this by \onlyLong{means of
Algorithm~\ref{alg:approx}, which essentially gathers} as much 
local information on $\Ga^s$ as possible, for every past round $s$.
Every process $p$ keeps track of its current graph approximation 
in variable\onlyLong{\footnote{We denote the value of a variable $v$ of process $p$ in
round $r$ (before the round $r$ computation finishes) as $v_p^r$; we usually
suppress the superscript when it refers to the current round.}}
$A_p$. Initially $A_p$ consists of process $p$ only, then $p$ broadcasts
and updates it in every round.
Ultimately, every process $p$ will use $A_p$ to determine whether it has 
been inside a vertex-stable root component for a sufficiently large interval.
\onlyLong{To
this end, Algorithm~\ref{alg:approx} provides predicate $\stableSCC(I)$, which returns true iff $p$ has been 
in the $I$-vertex
stable root component.}%
\onlyLong{
The edges of $A_p$ are labeled with a set of rounds constructed as
follows: Since $p$ can learn new information only via incoming messages, it
updates $A_p$, whenever $q \in \Timely_p^r$, by adding
$(q\edge{\{r\}}p)$ if $q$ is $p$'s neighbour for the first time, or
updating the label of the edge $(q\edge{U}p)$ to $(q\edge{U \cup
  \{r\}}p)$ (Lines \ref{line:addEdge1} and~\ref{line:addEdge2}).
Moreover, $p$ also receives $A_q$ from $q$ and uses this information
to update its own knowledge: 
The loop in Line~\ref{line:forloop} ensures that $p$ has an edge
$(v\edge{T\cup T'}w)$ for each $(v\edge{T'}w)$ in $A_q$, where 
$T$ is the set of rounds previously known to $p$.
}
Given $A_p$,
we will denote the information contained in $A_p$ about round $s$ by 
$A_p|s$.  More specifically, $A_p|s$ is the graph induced by
the set of edges
\onlyLong{\[}
\onlyShort{$}
    E_p|s=\set{e=(v\rightarrow w)\mid \exists T\supseteq\set{s}:
      (v\edge{T} w)\in A_p}.
\onlyShort{$}
\onlyLong{\]}

It is important to note that our Assumption~\ref{ass:window} is too
weak to guarantee that eventually the graph $A_p|s$ will ever exactly match
the actual $\Ga^s$ in some round $s$. In fact, there might be a process
$q$ that does not have any incoming links from other processes,
throughout the entire run of the algorithm. In that case, $q$ cannot learn
anything about the remaining network, i.e., $A_q$ will permanently be
the singleton graph. More generally, an $I$-vertex-stable root
component will not be able to acquire knowledge on the topology
outside $\R^I$ within $I$. 

\onlyLong{
To simplify the presentation, we have refrained from purging outdated
information from the network approximation graph. Our consensus algorithm 
only queries $\stableSCC$ for intervals that span at most the
last $4D+1$ rounds, i.e., any older information can safely be removed from 
the approximation graph, yielding a message complexity that is polynomial in
$n$.  }

\onlyLong{
\subsubsection{Proof of Correctness}

\begin{algorithm}[!h]
  \footnotesize
\caption{\em Local Network Approximation (Process~$p$)}
\label{alg:approx}
\setlinenosize{\footnotesize}
\setlinenofont{\tt}
\begin{algorithmic}[1]
\item[] Provides predicate $\stableSCC()$. %
\EMPTY
\item[] {\bf Variables and Initialization:}
\STATE $A_p:=\li{V_p,E_p}$ initially $(\set{p},\emptyset)$ \COMMENT{weighted digraph without multi-edges and loops}
\EMPTY

\item[] {\bf\boldmath Emit round $r$ messages:}
\STATE send $\msg{A_p}$ to all current neighbors
\EMPTY

\item[] {\bf\boldmath Round $r$: computation:}

\FOR{$q \in \Timely_p^r$ and $q$ sent message $\msg{A_q}$ in $r$}
  \IF{$\exists$ edge $e=(q\edge{T}p) \in E_p$}
    \STATE replace $e$ with $(q\edge{T'}p)$ in $E_p$ 
    where $T' \la T \cup \set{r}$ \label{line:addEdge1}
  \ELSE
    \STATE add $e:=(q\edge{\set{r}}p)$ to $E_p$ \label{line:addEdge2}
  \ENDIF
  \STATE $V_p \la V_p \cup V_q$
\ENDFOR
\EMPTY
\FOR{every pair of nodes $(p_i,p_j)\in V_p\times V_p$, $p_i\neq p_j$} \label{line:forloop}
  \IF{$T'=\bigcup\set{ S \mid \exists q\in\Timely_p^r\colon (p_i\edge{S}p_j)\in E_q} \neq \emptyset$}
    \STATE replace $(p_i\edge{T}p_j)$ in $E_p$ with $(p_i\edge{T\cup
    T'}p_j)$; add $(p_i\edge{T'}p_j)$ if no such edge exists
  \ENDIF
  \ENDFOR \label{line:forloopend}
\EMPTY
\EMPTY
\PRED{$\stableSCC(I)$}
  \STATE Let $A_p|s=(V_p^s,\set{(p_j\edge{T}p_i)\in E_p \mid s \in T})$
  \STATE Let $C_p|s$ be $A_p|s$ if it is strongly connected, or the empty
  graph otherwise.\label{line:Cps}
  \STATE return \textsc{true} \textit{iff} for all $s_1,s_2\in I$: 
  $V(C_p|s_1) = V(C_p|s_2) \ne \emptyset$
\ENDPRED
\EMPTY
\end{algorithmic}%
\end{algorithm}

The following Lemma~\ref{lem:AsubsetG} reveals that $A_p|t$ underapproximates 
$\G^t$ in a way that consistently
includes neighborhoods. Its proof uses a trivial invariant that is obvious from
the code, which says that
$A_p|t = \li{\{p\},\emptyset}$ at the end of every round $r<t$.

\begin{lemma} \label{lem:AsubsetG}
If $A_p|t$ contains $(v\edge{}w)$ at the end
of round $r$, then
\begin{enumerate}
\item[(i)] $(v\edge{}w)\in\G^t$, i.e.,
$A_p|t\subseteq\G^t$,
\item[(ii)] $A_p|t$ also contains $(v'\edge{}w)$ for every $v'\in 
\N_w^t \subseteq \G^t$.
\end{enumerate}
\end{lemma}
\begin{proof} 
We first consider the case where $r<t$, then at the end of round $r$
$A_p|t$ is empty, i.e., there are no edges in $A_p|t$. As the
precondition of the Lemma's statement is false, the statement is true.

For the case where $r\geq t$, we proceed by induction on $r$:

Induction base $r=t$: If $A_p|t$ contains $(v\edge{}w)$ at the end of 
round $r=t$,
it follows from $A_q|t = \li{\{q\},\emptyset}$ at the end of every 
round $r<t$, for every $q\in\Pi$, that $w=p$, since $p$ is the only
processor that can have added this edge to its graph approximation. 
Clearly, it did so only when $v\in \N_p^t$, i.e., 
$(v\edge{}w) \in \G^t$, and included also $(v'\edge{}w)$ for every 
$v'\in \N_p^t$ on that occasion. This confirms (i) and (ii).

Induction step $r\to r+1$, $r\geq t$: Assume, as our induction 
hypothesis, that
(i) and (ii) hold for any $A_q|t$ at the end of round $r$, 
in particular, for every
$q\in\N_p^{r+1}$. If indeed $(v\edge{}w)$ in $A_p|t$ at the end of 
round $r+1$, it must be contained in the union of round $r$
approximations
\[
  U=\left(A_p|t\right) \cup \left(\bigcup_{q\in \N_p^{r+1}} A_q|t\right)
\]
and hence in some $A_x|t$ ($x=q$ or $x=p$) at the end of round $r$.
Note that the edges (labeled $r+1$) added in round $r+1$ to $A_p$
are irrelevant for $A_p|t$ here, since $t < r+1$.

Consequently, by the induction hypothesis, $(v\edge{}w) \in \G^t$, thereby
confirming (i). As for (ii), the induction hypothesis also implies that 
$(v'\edge{}w)$ is also in this $A_x|t$. Hence, every such
edge must be in $U$ and hence in $A_p|t$ at the end of 
round $r+1$ as asserted.
\end{proof}

The next lemma shows that locally detecting a strongly connected component
$C_p|s \subseteq A_p|s$ (in Line~\ref{line:Cps} of Algorithm~\ref{alg:approx}) 
implies that $p$ is in the root component of round~$s$.
This result rests
on the fact that $A_p|s$ underapproximates $\G^s$ (Lemma~\ref{lem:AsubsetG}.(i)), but does so in a way that
never omits an in-edge at any process $q\in V(A_p|s)$ (Lemma~\ref{lem:AsubsetG}.(ii)).

\begin{lemma}\label{lem:Cpr2root}
If the graph $C_p|s$ (line~\ref{line:Cps}) with $s<r$ is non-empty in round $r$,
then $p\in\R^s$.
\end{lemma}

\begin{proof}
For a contradiction, assume that $C_p|s$ is non-empty (hence $A_p|s$ is
an SCC by Line~\ref{line:Cps}), but $p\not\in\R^s$. Since $p$ is always
included in any $A_p$ by construction and $A_p|s$ underapproximates
$\G^s$ by Lemma~\ref{lem:AsubsetG}.(i), this implies that $A_p|s$ cannot be
the root component of $\G^s$. Rather, $A_p|s$ must contain some process $w$ 
that has an in-edge $(v\edge{}w)$ in $\G^s$ that is not present in $A_p|s$.
As $w$ and hence some edge $(q\edge{s}w)$ is contained in $A_p|s$, because
it is an SCC, Lemma~\ref{lem:AsubsetG}.(ii) reveals that this is impossible.
\end{proof}

From this lemma and the description of predicate $\stableSCC(I)$ in
Algorithm~\ref{alg:approx}, we get
the following corollary.
\begin{corollary}\label{cor:stable2root}
  If the predicate $\stableSCC(I)$
  evaluates to $\true$ at process $p$ in round $r$, then $\forall s\in I$ where
  $s<r$, it holds that $p\in\R^{s}$.
\end{corollary}

The following Lemma~\ref{lem:root2Cpr} proves that, 
in a sufficiently
     long $I=[r,s]$ with a $I$-vertex-stable root component $\R^I$, 
every member $p$ of $\R^I$ detects an SCC for round $r$ (i.e., $C_p|r 
\neq \emptyset$) with a latency of at most $D$ rounds
(i.e., at the end of rounds $r+D$).
Informally speaking, together with Lemma~\ref{lem:Cpr2root}, it
asserts that if there is an $I$-vertex-stable root component $\R^I$ for
a sufficiently long interval $I$, then a process $p$ observes $C_p|r\neq 
\emptyset$ from the end of round $r+D$ on iff $p\in\R$. 

\begin{lemma} \label{lem:root2Cpr}
Consider an interval of rounds $I=[r,s]$, such that there is a \goodD\ 
$I$-vertex-stable root component $\R^I$ and assume $|I|=s-r+1 >D \geq D(\R^I)$. 
Then, from round $r+D$ onwards, we have $C_p|r=\R^I$,
     for every process in $p\in\R^I$.
\end{lemma}

\begin{proof}
Consider any $q\in\R^I$. At the beginning of round $r+1$, $q$ has
     an edge $(q'\edge{T}q)$ in its approximation graph $A_q$ with $r\in T$ iff
     $q'\in\N_q^r$. 
Since processes always merge all graph information from other
     processes into their own graph approximation, it follows from the definition of
     a \goodD\ $I$-vertex-stable root component in conjunction with the fact that
     $r+1\leq s-D+1$ that every $p\in\R^I$ has these
     in-edges of $q$ in its graph approximation by round $r+1+D-1$. Since $\R^I$
is a vertex-stable root-component, it is strongly connected without in-edges from
processes outside $\R^I$. Hence $C_p|r=\R^I$ from round $r+D$ on, as asserted.
\end{proof}

\begin{corollary} \label{cor:allrootdec}
Consider an interval of rounds $I=[r,s]$, with $|I|=s-r+1 > D\geq D(\R^I)$, 
such that there is a \goodD\ vertex-stable root component $\R^I$.
Then, from the end of round $s$ on, predicate $\stableSCC([r,s-D])$
evaluates to $\true$ at every process in $\R^I$.
\end{corollary}
}

\onlyLong{
\subsection{The Consensus Algorithm}
\label{sec:consalgo}
}

\onlyShort{
The underlying idea of our consensus algorithm is to use flooding to
     forward the largest input value to everyone. 
However, as Assumption~\ref{ass:window} does not guarantee
     bidirectional communication between every pair of processes,
     flooding is not sufficient: The largest proposal value could be
     hidden at a single process that never has outgoing edges. 
Therefore, we let ``privileged'' processes, namely, the ones in a
vertex-stable root component, try to impose their
     largest proposal values on the other processes. 
In order to do, so we use the well-known technique of locking a unique value. 
Processes only decide on their locked value once they are sure that
     every other process has locked this value as well. 
Since  Assumption~\ref{ass:window} guarantees that there will be one
root component such that the processes in the root component can communicate their locked value
to all other processes in the system they will eventually succeed.

\begin{theorem} \label{thm:consensus}
  Let $r_{ST}$ be the first round where Assumption~\ref{ass:window} holds.
  There is an algorithm that solves consensus by round $r_{ST}+4D+1$.
\end{theorem}
}

\begin{algorithm}
  \footnotesize
\caption{Solving Consensus; code for process $p$}
\label{alg:consensus}
\setlinenosize{\footnotesize}
\setlinenofont{\tt}
\begin{algorithmic}[1]
\STATE {Simultaneously run Algorithm~\ref{alg:approx}.}
\EMPTY
\item[] {\bf Variables and Initialization:}
\STATE $x_p \in \mathbb{N}$ initially $v_p$ 
\STATE $locked_p, decided_p \in \set{\texttt{false},\texttt{true}}$ initially
  $\texttt{false}$ 
\STATE $lockRound_p \in \mathbb{Z}$ initially $0$  
\EMPTY

\item[] {\bf\boldmath Emit round $r$ messages:}
\IF{$decided_p$}
\STATE send $\msg{\textsc{decide},x_p}$ to all neighbors         \label{line:sendDec}
\ELSE
  \STATE send $\msg{lockRound_p,x_p}$ to all neighbors
\ENDIF
\EMPTY

\item[] {\bf\boldmath Round $r$ computation:} %
\IF{\textbf{not} $decided_p$} \label{line:secondIf}
  \IF{received $\msg{\textsc{decide},x_q}$ from any neighbor $q$}  \label{line:firstIf}
    \STATE $x_p \la x_q$
    \STATE decide on $x_p$ and set $decided_p \la \texttt{true}$\label{line:decOther}
  \EMPTY

  \ELSE[{$p$ only received $\msg{lock_q,x_q}$ messages (if any)}:]
  \STATE $(lockRound_p,x_p) \la \max\set{ (lock_q,x_q) \mid q \in
    \Timely_p^r \cup \{p\}}$ \COMMENT{lexical order in $\max$} \label{line:updateXp}
  \IF{$\stableSCC([r-D-1,r-D])$\label{line:if-G}}
    \IF{$(\text{\textbf{not} $locked_p$})$}
      \STATE $locked_p\la \texttt{true}$
      \STATE $lockRound_p \la r$ \label{line:lock}
    \ELSE 
      \IF{$\stableSCC([lockRound_p,lockRound_p+D])$\label{line:if-lock-heard}}
        \STATE decide on $x_p$ and set $decided_p \la \texttt{true}$
      \label{line:decideOwn} \label{line:decOwn}
      \ENDIF
    \ENDIF
  \ELSE[{$\stableSCC([r-D-1,r-D])$ returned $\false$}]
    \STATE $locked_p\la \texttt{false}$
  \ENDIF
  \ENDIF
\ENDIF
\end{algorithmic}%
\end{algorithm}

\onlyLong{
As we explained before our consensus algorithm is built atop of the network approximation
algorithm.  More specifically, we rely on Corollary~\ref{cor:stable2root}
and use the 
$\stableSCC$ predicate provided by Algorithm~\ref{alg:approx} to 
detect that a process is currently (potentially) in a vertex-stable root component.
To be able to do so, round $r$ of 
Algorithm~\ref{alg:approx} is executed before round $r$ of 
Algorithm~\ref{alg:consensus}, and messages sent in round $r$ by both
algorithms are packed together in a single message.
Since Corollary~\ref{cor:allrootdec} revealed that
$\stableSCC$ has a delay of up to $D$ rounds to detect that a
process is in the vertex-stable root component of some interval
of rounds, however, our algorithm (conservatively)
looks back $D$ rounds in the past for locking in order to ensure that
the information returned by $\stableSCC$ is reliable.

In more detail, Algorithm~\ref{alg:consensus} proceeds as follows: 
Initially, no process has locked a value, that is $locked=\false$ and
     $lockRound_p=0$. 
Processes try to detect whether they are privileged by evaluating the
     condition in Line~\ref{line:firstIf}. 
When this condition is true in some round $\ell$, they lock the
     current value (by setting $locked_p=\true$ and $lockRound$ to the
     current round), unless $locked_p$ is already $\true$. 
Note that our locking mechanism does not actually protect the value
     against being overwritten by a larger value being also locked in $\ell$;
it locks out only those values that have older locks $l<\ell$.

When the process $m$ that had the largest value in the root component
     of round $\ell$ detects that it has been in a vertex-stable root component
     in all rounds $\ell$ to $\ell+D$ (Line~\ref{line:if-lock-heard}), it
     can decide on its current value. 
As all other processes in that root component must have had $m$'s value imposed
     on them, they can decide as well.
     After deciding, a process stops participating in the
     flooding of locked values, but rather (Line~\ref{line:sendDec})
     floods the network with $\msg{\textsc{decide},x}$. 
Since the time window guaranteed by
     Assumption~\ref{ass:window} is large enough to allow every process
     to receive this message, all processes will eventually decide.

\subsubsection{Proof of Correctness}
\begin{lemma}[Validity] \label{lem:validity}
  Every decision value is the input value of some process.
\end{lemma}

\begin{proof}
Processes decide either in Line~\ref{line:decOther} or in
     Line~\ref{line:decOwn}. 
When a process decides via the former case, it has received a
     $\msg{\textsc{decide},x_q}$ message, which is sent by $q$ iff $q$
     has decided on $x_q$ in an earlier round. 
In order to prove the theorem, it is thus sufficient to show that
     processes can only decide on some process' input value when they
     decide in Line~\ref{line:decOwn}, where they decide on their
     current estimate $x_p$. 
Let the round of this decision be $r$. 
This value is either $p$'s initial value, or was updated in some round
     $r'\le r$ in Line~\ref{line:updateXp} from a value received by
     way of one of its neighbors' $\msg{lockRound,x}$ message. 
In order to send such a message, $q$ must have had $x_q=x$ at the
     beginning of round $r'$, which in turn means that $x_q$ was
     either $q$'s initial value, or $q$ has updated $x_q$ after
     receiving a message in some round $r_q<r$.  
By repeating this argument, we will eventually reach a process
that sent its initial value, since no process can have updated its
     decision estimate prior to the first round.
\end{proof}

The following Lemma~\ref{lem:assorted-properties} states a number of 
properties maintained by our algorithm when the first process $p$ has 
decided.
Essentially, they say that there has been a vertex-stable root component 
for at
least $2D+1$ rounds centered around the lock round $\ell$ (but not earlier),
and asserts that all processes in that root component chose the same lock round $\ell$.

\begin{lemma}\label{lem:assorted-properties}
Suppose that process $p$ decides in round $r$, no decisions
     occurred before $r$, and $\ell=lockRound_p^r$,\footnote{We denote the value of a variable $v$ of process $p$ in
round $r$ (before the round $r$ computation finishes) as $v_p^r$; we usually
suppress the superscript when it refers to the current round.} then 
     \begin{itemize}
     \item[(i)] $p$ is
     in the vertex-stable root component $\R^I$ with
     $I=[\ell-D-1,\ell+D]$, 
     \item[(ii)] $\ell+D\le r\le\ell+2D$,  
     \item[(iii)]$\R^I\ne\R^{\ell-D-2}$, and  
     \item[(iv)] all processes in $\R^I$
     executed Line~\ref{line:lock} in round $\ell$, and no process
in $\Pi\setminus{\R^I}$ can have executed Line~\ref{line:lock} in
a round $\geq \ell$.
     \end{itemize}
\end{lemma}

\begin{proof}
Item (i) follows since Line~\ref{line:if-G} has been continuously
      $\true$ since round $\ell$ and from Lemma~\ref{lem:Cpr2root}.
As for item (ii), $\ell+D\le r$ follows from the requirement of
      Line~\ref{line:if-lock-heard}, while $r\le\ell+2D$ follows from
      (i) and the fact that by Lemma~\ref{lem:root2Cpr} the
      requirement of Line~\ref{line:if-lock-heard} cannot be, for the
      first time, fullfilled strictly after round $\ell+2D$. 
From Lemma~\ref{lem:root2Cpr}, it also follows that if
      $\R^{\ell-D-2}=\R^I$, then the condition in Line~\ref{line:if-G}
      would return true already in round $\ell-1$, thus locking in round
      $\ell-1$. 
Since $p$ did not lock in round $\ell-1$, (iii) must hold. Finally, 
from (i), (iii), and Lemma~\ref{lem:root2Cpr}, it follows that
      every other process in $\R^I$ also has
      $\stableSCC([\ell-D-1,\ell-D])=\true$ in round $\ell$. 
Moreover, due to (iii), $\stableSCC([\ell-1-D-1,\ell-1-D])=\false$ 
in round $\ell-1$, which causes
      all the processes in $\R^I$ (as well as those in $\Pi\setminus{\R^I}$)
to set $lockRound$ to 0. 
Since $\stableSCC([\ell'-D-1,\ell'-D])$ cannot become true for any 
$\ell'\geq \ell$ at a process $q\in \Pi\setminus\R^I$, as $C_q|r=\emptyset$
for any $r\in I$, (iv) also holds.
\end{proof}

The following Lemma~\ref{lem:highlander} asserts that if a process
decides, then it has successfully imposed its proposal value on all other 
processes.

\begin{lemma}[Identical proposal values] \label{lem:highlander}
Suppose that process $p$ decides in Line~\ref{line:decideOwn} in round $r$ and
that no other process has executed Line~\ref{line:decideOwn} before $r$.
Then, for all $q$, it holds that $x_q^r = x_p^r$.
\end{lemma}

\begin{proof}
Using items (i) and (iv) in Lemma~\ref{lem:assorted-properties}, we can
     conclude that $p$ was in the vertex-stable root component $\R$ of
     rounds $\ell=lockRound_p^r$ to $\ell+D$ and that all processes  in
     $\R$ have locked in round $\ell$. 
Therefore, in the interval $[\ell,\ell+D]$, $\ell$ is the maximal
     value of $lockRound$. 
More specifically, all processes $q$ in $\R$ have $lockRound_q=\ell$,
     whereas all processes $s$ in $\Pi\setminus\R$ have
     $lockRound_s<\ell$ during these rounds by Lemma~\ref{lem:assorted-properties}.(iv).
Let $m\in \R$ have the largest proposal value
     $x_m^\ell=x_{max}$ among all processes in $\R$. 
Since $m$ is in $\R$, there is a causal chain of length at most $D$
     from $m$ to any $q\in\Pi$. 
Since no process executed Line~\ref{line:decideOwn} before round $r$, no
     process will send $\textsc{decide}$ messages in $[\ell,\ell+D]$. 
Thus, all processes continue to execute the update rule of
     Line~\ref{line:updateXp}, which implies that $x_{max}$ will
     propagate along the aforementioned causal path to $q$.
\end{proof}

\begin{theorem} \label{thm:consensus}
  Let $r_{ST}$ be the first round where Assumption~\ref{ass:window} holds.
  Algorithm~\ref{alg:consensus} in conjunction with Algorithm~\ref{alg:approx} solves
  consensus by round $r_{ST}+4D+1$.
\end{theorem}
\begin{proof}
  Validity holds by Lemma~\ref{lem:validity}. Considering Lemma~\ref{lem:highlander}, 
  we immediately get Agreement: Since the first process $p$ that decides must
  do so via Line~\ref{line:decOwn}, there are no other proposal values left
  in the system.

  Observe that, so far, we have not used the liveness part of 
  Assumption~\ref{ass:window}.  In
  fact, Algorithm~\ref{alg:consensus} is always safe in the sense that Agreement
  and Validity are not violated, even if there is no vertex-stable root
  component. 

  We now show the Termination property. By
  Corollary~\ref{cor:allrootdec}, we know that every process in $p\in\R$
  evaluates the predicate $\stableSCC([r_{ST},r_{ST}+1])=\true$ in
  round $\ell=r_{ST}+D+1$, thus locking in that round.
  Furthermore, Assumption~\ref{ass:window} and Corollary~\ref{cor:allrootdec}
  imply that at the latest in round $d=\ell+2D$ every process $p\in\R$ will evaluate the
  condition of Line~\ref{line:if-lock-heard} to $\true$ and thus
  decide using Line~\ref{line:decOwn}.
  Thus, every such process $p$ will send out a message
  $m=\msg{\textsc{decide},x_p}$. 
By the definition of $D$ and Assumption~\ref{ass:window}, we know that
the round $d$ network causal diameter satisfies $D^{d}(\Pi)\le D$, such that every $q\in\Pi$ will
receive the $\textsc{decide}$ message at the latest in round $d+D=\ell+3D=r_{ST}+4D+1$.
\end{proof}

\subsection{Improved Time Complexity}
\label{sec:improved}

We now discuss how to guarantee a smaller causal network diameter, by considering
networks with sufficient expansion properties.
Recall that an undirected graph $\G$ is an 
\emph{$\alpha$-vertex expander} if, for all sets $S \subset V$ of size $\le 
|\G|/2$, it holds that $\frac{|\N(S)|}{|S|} \ge \alpha$, where $\N(S)$ is
the set of neighbors of $S$ in $\G$, i.e., those nodes in
$V(\G)\setminus S$ that have a neighbor in $S$.
Such graphs exist and can be constructed explicitly (cf.\
\cite{HLW2006:Expander}).

Interpreting the above statement for undirected expanders in the
     context of communication graphs allows to reason both about the
     incoming and the outgoing transmissions of a process.
For defining our expander property for directed communication networks
we therefore consider (for a vertex/process set $S$ and a round $r$)
both the set $\N^r_+(S)$ of nodes outside of $S$ that are reachable from $S$ and 
the set of nodes $\N^r_-(S)$ that can reach $S$ in $r$.

\begin{assumption}[Expander Topology] \label{ass:fast}
There is a fixed constant $\alpha$ and fixed set $\R$ such that the following 
conditions hold for all sets $S \subseteq V(\G^r)$:
\begin{compactenum}
\item[(a)] If $|S|\le |\R|/2$ and $S\subseteq\R$, then
  $\frac{|\N^r_+(S) \cap \R|}{|S|} \ge \alpha$ and $\frac{|\N^r_-(S) 
  \cap \R|}{|S|} \ge \alpha$.
\item[(b)] If $|S|\le n/2$ and $\R\subseteq S$, then
  $\frac{|\N^r_+(S)|}{|S|} \ge \alpha$.
\item[(c)] If $|S|\le n/2$ and $\R \cap S = \emptyset$, then 
  $\frac{|\N^r_-(S)|}{|S|} \ge \alpha$.
\end{compactenum}
\end{assumption}

Our next theorem shows that (1) the assumption of a expander topology
     for communication graphs does not contradict our previous
     assumption about root components and that (2) these expander
     topologies also guarantee the causal diameter and thus the time
     complexity of our algorithm to be in $O(\log n)$.

\begin{lemma}
  There are sequences of graphs where Assumptions~\ref{ass:window} and 
  \ref{ass:fast} are satisfied simultaneously and, for any such run,
  there is an interval $I$ during which there exists an $O(\log 
  n)$-bounded $I$-vertex stable root component.
\end{lemma}

\begin{proof}
We will first argue that \emph{directed} graphs exist that simultaneously 
satisfy Assumptions~\ref{ass:window} and \ref{ass:fast}.
Consider the simple \emph{undirected} graph $\bar{\U}$ that is the 
union of an $\alpha$-vertex expander on $\R^I$, and
an $\alpha$-vertex expander on $V(\G^r)$. 
We turn $\bar{U}$ into a directed graph by replacing every edge $(p,q)\in E(\bar{\U})$ with oriented 
directed edges $p\ra q$ and $q\ra p$. This guarantees Properties
(a)-(c). In order to guarantee the existence of exactly one root
component, we drop all directed edges pointing to
$\R$ from the remaining graph, i.e., we remove all edges  $p\ra q$
where $p\in\R$ and $q\not\in\R$, which leaves Properties (a)-(c)
intact and makes the $\R$ from Assumption~\ref{ass:fast} the single
root component of the graph.
We stress that the actual topologies chosen by the adversary might be 
quite different from this construction, which merely serves us to show 
the existence of such graphs.

It is also worth mentioning that Assumption~\ref{ass:window}
only requires the set of vertices in $\R^J$ to remain unchanged but
the topology can change arbitrarily. With also assuming
Assumption~\ref{ass:fast} this does not change. However, as we show
next the (per round) expander topology is strong enough to guarantee a network causal diameter in 
$O(\log n)$.

For $i\ge1$, let $\P_i\subseteq\R^I$ be the set of processes $q$ in $\R^I$ 
such that $(p\ltedge{b[i]}q)$, and $\P_0=\{p\}$.
We first show that $D^I(\R^I) \in O(\log n)$.
The result is trivial if $|\R^I|\in O(\log n)$, thus assume that 
$|\R^I|\in\Omega(\log n)$ and consider some process $p \in \R^I$.
For round $b$, Property (a) yields $|\P_{1}| \ge |\P_{0}|(1+\alpha)$.
In fact, for all $i$ where $|\P_{i}|\le |\R^I|/2$, we can apply Property (a) to 
get $|\P_{i + 1}| \ge |\P_{i}|(1+\alpha)$, hence $|\P_{i}| \ge
\min\{(1+\alpha)^i,|\R^I|/2\}$. 
Let $k$ be the smallest value such that
$|(1+\alpha)^{k}|>|\R^I|/2$, which guarantees that $|\P_{k}|>|\R^I|/2$.
That is, $k = \left\lceil\frac{\log(|\R^{I}|/2)}{\log(1+\alpha)}\right\rceil 
\in O(\log n)$.
Now consider any $q\in\R^I$ and define $\Q_{i-1}\subset\R^I$ as 
the set of nodes that causally influence the set $\Q_i$ in round $b+i$, for 
$\Q_{2k+1}=\{q\}$.
Again, by Property (a), we get $|\Q_{i-1}| \ge |\Q_i|(1+\alpha)$, so $|\Q_{2k-i}| \ge 
\max\{(1+\alpha)^{i},|\R|/2\}$. From the definition of $k$ above we
thus have $|\Q_{k}|>|\R^I|/2$.
Since $\P_{k} \cap \Q_{k} \ne \emptyset$, it follows that every $p \in 
\R^I$ influences every $q\in \R^I$ within $2k\in O(\log n)$
rounds. Note that while we have shown this for $\dist_b(p,q)$ only,
this is also valid for any $\dist_r(p,q)$ with $r<e-2k$.

Finally, to see that this guarantees $D$-boundedness with $D\in O(\log
n)$, we use Properties (b) and (c) similarly to the case above.
For any round $r \in [b,e-2k']$, we know by (b) that any process $p\in \R^I$ has 
influenced at least $n/2$ nodes by round $r+k'$ where $k' = \rceil 
\log_{1+\alpha}(n/2)\lceil \in O(\log n)$ by arguing as for the $\P_i$
sets above.
Now (c) allows us to reason along the same lines as for the sets 
$\Q_{i-1}$ above. That is any $q$ in round $r+2k'$ will be influenced 
by at least $n/2$ nodes. Therefore, any $p$ will influence every $q\in\Pi$ by round 
$r+2k'$, which completes the proof.
\end{proof}

\begin{corollary}
  Suppose that Assumptions~\ref{ass:window} and \ref{ass:fast} hold.
  Then, running Algorithms~\ref{alg:approx} and \ref{alg:consensus} solves 
  consensus by round $r_\text{ST}+O(\log n)$.
\end{corollary}
}
\section{Impossibilities and Lower Bounds} \label{sec:imposs}

\onlyShort{In this section, we will present a number of results that
  show that our basic Assumption~\ref{ass:window}, }%
\onlyLong{In this section, we will prove that our basic
  Assumption~\ref{ass:window}, }%
in particular,
the existence of a stable window (of a certain minimal size) and the 
knowledge of an upper bound $D$ on the causal network
diameter, are crucial for making consensus solvable. Moreover, we will show
that it is not unduly strong, as many problems considered in distributed systems 
in general (and dynamic networks in particular) remain unsolvable.

\onlyLong{
First, we relate our Assumption~\ref{ass:window} to the classification 
of~\cite{CFQS11:TVG}.  Lemma~\ref{lem:classification}
reveals that it is stronger than one of the two weakest classes, but also
weaker than the next class.

\begin{lemma}[Properties of root components]\label{lem:classification}
Assume that there is at most one root component $\R^r$ in every $\Gr$, $r>0$.
Then, (i) there is at least one process $p$ such that $\dist_1(p,q)$ is
finite for all $q\in\Pi$, and this causal distance is in fact at most $n(n-2)+1$. 
Conversely, for $n>2$, the adversary can choose topologies 
where (ii) no process $p$ is causally influenced by all other processes 
$q$, i.e., ${\not\exists p}\ \forall q\colon \dist_1(q,p) < \infty$.
\end{lemma}
\begin{proof}
  Since we have infinitely many rounds in a run but only finitely many
  processes, there is at least one process $p$ in
  $\R^r$ for infinitely many $r$. Let $r_1,r_2,\dots$ be this sequence
  of rounds. Moreover, let $\P_0=\set{p}$, and define for each $i>0$ the set
  $\P_i = \P_{i-1} \cup \{q: \exists q'\in \P_{i-1}: q' \in \N_q^{r_i}\}$.

Using induction, we will show that $|\P_k| \geq \min\{n,k+1\}$ for $k\geq 0$. 
Consequently, by the end of round $r_{n-1}$ at latest, $p$ will have causally influenced all
processes in $\Pi$.
Induction base $k=0$: $|\P_0| \geq \min\{n,1\}=1$ follows immediately from
$\P_0=\set{p}$.
Induction step $k \to k+1$, $k\geq 0$:  First assume that already
$|\P_k|=n \geq \min\{n,k+1\}$; since $|\P_{k+1}|\geq
|\P_k|=n \geq \min\{n,k+1\}$, we are done.
  Otherwise, consider round $r_{k+1}$ and $|\P_k|<n$: Since $p$ is in $\R^{r_{k+1}}$, there is a path from
  $p$ to any process $q$, in particular, to any process $q$ in
  $\Pi\setminus\P_k \neq \emptyset$. Let $(v\rightarrow w)$ be an edge on such a
  path, such that $v\in\P_k$ and $w\in\Pi\setminus\P_k$. Clearly, the existence of
  this edge implies that $v\in\N_{w}^{r_{k+1}}$ and thus $w\in\P_{k+1}$.
  Since this implies $|\P_{k+1}| \geq |\P_k| + 1 \geq k+1 + 1 = k+2 = \min\{n,k+2\}$
by the induction hypothesis,
  we are done.

Finally, at most $n(n-2)+1$ rounds are needed until all processes $q$
have been influenced by $p$, i.e., that $r_{n-1}\leq n(n-2)+1$:
A pigeonhole argument reveals that at
least one process $p$ must have been in the root component for $n-1$ times
after so many rounds: If every $p$ appeared at most $n-2$ times, we
can fill up at most $n(n-2)$ rounds. By the above result, this is
enough to secure that some $p$ influenced every $q$.

  The converse statement (ii) follows directly from
  considering a static star, for example, i.e., a communication graph where there is one central
     process $c$, and for all $r$, $\Gr=\li{\Pi,\set{c\ra q|
     q\in\Pi\setminus\set{c}}}$. 
Clearly, $c$ cannot be causally influenced by any other process, and $q\not\lt q'$
for any $q,q'\neq q \in \Pi\setminus\set{c}$. On the other hand,
this topology satisfy Assumption~\ref{ass:window}, which
includes the requirement of at most one root component per round. 
\end{proof}
}

\onlyLong{
Next, we examine the solvability of several broadcast problems under 
Assumption~\ref{ass:window}.}
Although there is a strong bond between some of these problems and consensus in 
traditional settings, they are \emph{not} implementable under our 
assumptions---basically, because there is no guarantee of
(eventual) bidirectional communication. 

\onlyLong{
We first consider reliable broadcast,
which requires that when a correct process broadcasts
$m$, every correct process eventually delivers $m$. Suppose
the adversary chooses the topologies $\forall r:\ 
\Gr=\li{\set{p,q,s},\set{p\ra q,q\ra s}}$, which matches Assumption~\ref{ass:window}.
Clearly, $q$ is a correct process in our model. Since $p$ never receives a message from $q$,
$p$ can trivially never deliver a message that $q$ broadcasts.

We now turn our attention to the various problems considered
in~\cite{kuhn+lo:dynamic}, which are all impossible to solve
under Assumption~\ref{ass:window}.
More specifically, we return to the static star considered in the
proof of Lemma~\ref{lem:classification}.
Clearly, the local history of any process is independent of the size
     $n$.
Therefore,  the problems
of counting, $k$-verification, and $k$-committee election are
     all impossible. 
For the token dissemination problems, consider that there is a token
     that only $p\ne c$ has. 
Since no other process ever receives a message from $p$, token dissemination is
     impossible.
}
\begin{theorem} \label{thm:impossibleProblems}
Suppose that Assumption~\ref{ass:window} is the only restriction on the 
adversary in our model.
Then, neither \emph{reliable broadcast}, \emph{atomic broadcast}, nor 
\emph{causal-order broadcast} can be implemented.
Moreover, there is no algorithm that solves
\emph{counting}, \emph{$k$-verification}, \emph{$k$-token
dissemination}, \emph{all-to-all token dissemination}, and 
\emph{$k$-committee election}.
\end{theorem}

\onlyLong{\subsection{Knowledge of a Bound on the Network Causal Diameter}
\label{sec:impossDiameter}
}

\begin{theorem}%
  \label{thm:impossDiameter}
Consider a system where Assumption~\ref{ass:window} holds and
suppose that processes do not know an upper bound $D$ on the network causal
diameter (and hence do not know $n$). Then, there is no 
algorithm that solves consensus.
\end{theorem}
\onlyLong{
\begin{proof}
Assume for the sake of a contradiction that there is such an algorithm $\A$.
For $v \in \{0,1\}$, let $\alpha(v)$ be a run of $\A$ on a communication graph $G$
that forms a static directed line 
rooted at process $p$
and ending in process $q$.
Process $p$ has initial value $v$, while all other processes have
initial value $0$.
Clearly, algorithm $\A$ must allow
$p$ to decide on $v$ by the end of round $\kappa$, where $\kappa$ is a constant
(independent of $D$ and $n$; we assume that $n$ is large enough to guarantee $n-1 >\kappa$).
Next, consider a run $\beta(v)$ of $\A$ that has the same initial
states as $\alpha(v)$, and communication graphs $(H_r)_{r>0}$
that, during rounds $[1,\kappa]$, are also the same as in $\alpha(v)$
(we define what happens after round $\kappa$ below).
In any case, since $\alpha(v)$ and $\beta(v)$ are indistinguishable for $p$
until its decision round $\kappa$, it must also decide $v$ in $\beta(v)$
at the end of round $\kappa$.

However, since $n > \kappa+1$, $q$ has not been causally influenced by $p$
by the end of round $\kappa$. Hence, it has the same state $S_p^{\kappa+1}$
both in $\beta(v)$ and in $\beta(1-v)$. As a consequence, it cannot
have decided by round $\kappa$: If $q$ decided $v$, it would violate 
agreement with $p$ in $\beta(1-v)$. Now assume that run $\beta(.)$
is actually such that the stable window occurs later than round $\kappa$, 
i.e., $r_{ST}=\kappa+1$, and that the adversary just reverses
the direction of the line then: For all $H^{k}$, $k\geq \kappa+1$, $q$ is 
the root and $p$ is leaf in the resulting topology.
Observe that the resulting $\beta(v)$ still satisfies
Assumption~\ref{ass:window}, since $q$ itself forms the only root component.  
Now, $q$ must eventually decide on some value $v'$ in some later round $\kappa'$, 
but since $q$ has been in the 
same state at the end of round $\kappa$ in both $\beta(v)$
and $\beta(1-v)$, it is also in the same state in round $\kappa'$
in both runs. Hence, its decision contradicts the decision of $p$ in 
$\beta(1-v')$.
\end{proof}

\subsection{Impossibility of Leader Election with Unknown Network Causal 
Diameter} \label{sec:diameter}

We now use a more involved indistinguishability argument to show
that a weaker problem than consensus, namely, leader election is also
impossible to solve in our model. The classic leader election problem 
(cf.\ \cite{Lyn96})
assumes that, eventually, exactly one process irrevocably elects itself as leader (by 
entering a special \textsc{elected} state) and every other process elects 
itself as non-leader (by entering the \textsc{non-elected} state).  
Non-leaders are not required to know the process id of the leader.

Whereas it is easy to achieve leader election in our model when consensus is
solveable, by just reaching consensus on the process ids in the system, the opposite
is not true: Since the leader elected by some algorithm need not be in the root component
that exists when consensus terminates, one cannot use the leader to 
disseminate a common value to all processes in order to solve
consensus atop of leader election.

\begin{theorem}
  \label{thm:impossDiameterLE}
Consider a system where Assumption~\ref{ass:window} holds and
suppose that processes do not know an upper bound $D$ on the network causal
diameter (and hence do not know $n$). Then, there is no algorithm that solves leader election.
\end{theorem}

\begin{proof}
We assume that there is an  algorithm $\A$ which solves the problem. 
Consider the execution $\alpha(w,k)$ of $\A$ in a static unidirectional 
chain of $k$ processes, headed by process $p$ with id $w$: Since $p$ has only a single out-going edge and does not
know $n$, it cannot know whether it has neighbors at all. Since it might
even be alone in the single-vertex graph consisting of $p$ only, it
must elect  itself as leader in any $\alpha(w,k)$, after some $T(w,k)$ rounds (we do 
not restrict $\A$ to be time-bounded). Again, note that $T(w,k)$ does not
depend on $k$, and so in fact $p$ will elect itself after the same
$T(w)$ rounds in all $\alpha(w,k)$, $k\ge1$.

Let $w$ and $z$ be two arbitrary different process ids, and let $T(w)$ resp.\
$T(z)$ be the termination times in the executions $\alpha(w,k)$
resp.\ $\alpha(z,k')$; let $T=\max\{T(w),T(z)\}$.

We now build a system consisting of $n=2T+3$ processes. To do so we 
assume a chain $\G_p$ of $T+1$ processes headed by $p$ (with id $w$)
and ending in process $t$, a second chain $\G_q$ of $T+1$ processes headed
by $q$ (with id $z$) and ending in process $s$, and the process $r$.

Now consider an execution $\beta$, which proceeds as follows: 
For the first $T$ rounds, the communications graph is the
     unidirectional ring created by connecting the above chains with
     edges $s\ra p$, $t\ra r$ and $r\ra q$; 
its root component clearly is the entire ring. 
Starting from round $T+1$ on, process $r$ forms the single vertex root
     component, which feeds, through edges $r\ra q$ and $r\ra t$ the
     two chains $\G_q$ and $\bar{\G}_p$, with $\bar{\G}_p$ being $\G_p$ with all
     edges reversed. 
Note that, from round $T+1$ on, there is no edges connecting processes
     in $\G_p$ with those in $\G_q$ or vice versa.

Let $\ell$ be the process that is elected leader in $\beta$.
We distinguish 2 cases: 

\begin{enumerate}
\item If $\ell \in \G_q \cup \{r\}$, then
consider the execution $\beta_p$ that is exactly like $\beta$, except
that there is no edge $(s\ra p)$ during the first $T$ rounds:
$p$ with id $w$ is the single root component here. Clearly, for
$p$, the execution $\beta_p$ is indistinguishable from $\alpha(w,2T+3)$
during the first $T(w)\leq T$ rounds,
so it must elect itself leader. However, since no process in $\G_q \cup \{r\}$ (including $t=\ell$)
is causally influenced by $p$ during the first $T$ rounds, all
processes in $\G_q\cup\{r\}$ has the same
state after round $T$ (and all later rounds) in $\beta_p$ as in $\beta$.  
Consequently, $\ell$
also elects itself leader in $\beta_p$ as it does in $\beta$, which is a 
contradiction.
\item On the other hand, if $\ell \in \G_p$, we consider the execution
$\beta_q$, which is exactly like $\beta$, except
that there is no edge $(r\ra q)$ during the first $T$ rounds:
$q$ with id $z$ is the single root component here. Clearly, for
$q$, the execution $\beta_q$ is indistinguishable from $\alpha(z,\G_q)$
during the first $T(z)\leq T$ rounds,
so it must elect itself leader. However, since no process $t$ in $\G_p \cup \{r\}$ (including $t=\ell$)
is causally influenced by $q$ during the first $T$ rounds, $t$ has the 
same
state after round $T$ (and all later rounds) in $\beta_q$ as in $\beta$.  
Consequently, $\ell$
also elects itself leader $\beta_q$ as it does in $\beta$, which is again 
a contradiction.
\end{enumerate}
This completes the proof of Theorem~\ref{thm:impossDiameterLE}.
\end{proof}
}

\onlyLong{
\subsection{Impossibility of consensus with too short intervals}

Our goal in this section is to show}
\onlyShort{We now state a result that shows} that it is necessary to have
root components that are vertex stable long enough to flood the
network. That is, with respect to Assumption~\ref{ass:window}, we need $I$ to be
at least $D$.
To this end, we first introduce the following alternative Assumption~\ref{ass:shorter},
which requires a window of only $D$: 
\begin{assumption}\label{ass:shorter}
For any round $r$, there is exactly one root component $\R^r$ in
     $\Gr$. 
Moreover, there exists a $D$ and an interval of rounds $I = [r_{ST},
     r_{ST}+D]$, such that there is an $I$-vertex stable root
     component $\R^I$, such that $D^I\le D$.
\end{assumption}

In order to show that Assumption~\ref{ass:shorter} is necessary, we further
shorten the interval: Some process could possibly 
not be reached within $D-1$ rounds, but would be reached if the
interval was $D$ rounds. Processes could hence \emph{withold} information
from each other, which causes consensus to be impossible \cite{SWK09}.
\onlyLong{To simplify the proofs, we consider a stronger variant, where 
  there is exactly one such process $q$, that is not reached within
$D-1$ rounds, from any process in $\R^I$. Note that, since it is
exactly one, we have that in executions where the interval is actually
$D$ rounds, it will be reached.
Thus we consider the following assumption:

\begin{assumption}\label{ass:too-short}
For any round $r$, there is exactly one root component $\R^r$ in
     $\Gr$. 
Moreover, there exists a $D$ and an interval of rounds $I = [r_{ST},
     r_{ST}+D-1]$, such that there is an $I$-vertex stable root
     component $\R^I$, and there exists a unique $q\in\Pi$ such that
     $\forall p\in\R^I, \forall r\in I: \dist_r(p,q)\ge D$, while for
     all $q'\in\Pi\setminus\set{q}$ we have 
     $\forall p\in\R^I, \forall r\in I:\dist_r(p,q')\le D-1$.
\end{assumption}

Note that, those executions that fulfill Assumption~\ref{ass:shorter}
(or even Assumption~\ref{ass:window} also fulfill Assumption~\ref{ass:too-short}.

In order to simplify the following proof we assume that the adversary has to fix the
start of $I$ and the set of processes in the root component $\R^r$ of
every round $r$ before the beginning of the execution (but given the
initial values). 
Note that this does not strengthen the adversary, and hence does not 
weaken our impossibility result. 
This is due to the fact that for 
deterministic algorithms the whole execution depends only on the
initial values and the sequence of round graphs so in this case the
adversary could simulate the whole execution and determine graphs
based on the simulation. 

\begin{lemma}\label{lem:neighbours-are-bivalent}
Assume some fixed $I$ and $\R^I$, such that
     Assumption~\ref{ass:too-short} holds. 
If two univalent configurations $C'$ and $C''$ at the beginning of
     round $r$ differ only in the state of one process $p$, they
     cannot differ in valency.
\end{lemma}

\begin{proof}
The proof proceeds by assuming the contrary, i.e., that $C'$ and $C''$
have different valency. We will then apply the same sequence of round
graphs to extend the execution prefixes that led to $C'$ and $C''$ to
get two different runs $e'$ and $e''$. It suffices to show that
there is at least one process $q$ that cannot distinguish $e'$ from
$e''$: This implies that $q$ will decide on the same value in both
executions, which contradicts the assumed different valency of $C'$
and $C''$. 

Our choice of the round graphs depends on the following two cases: 
(i) $p$ is in $\R^r$ and $r\in I$ or (ii) otherwise. 
In the second case, we assume that the adversary choses
$\R^s=\set{q}\ne\set{p}$ for all rounds $s\geq r$. We can thus consider 
     a sequence of graphs $\G^{s}$, for $s\ge r$, such that
     $\dist_s(q,p)=D$ which
obviously fulfills Assumption~\ref{ass:too-short}. 
But for $q$ (and all other processes except $p$), $e'$ and $e''$
are indistinguishable.

In case (i), by our Assumption~\ref{ass:too-short}, there is some $q$,
such that the information that $p$ sends in round $r$ does not arrive
at some specific $q$ within $I=[a,b]$. 
Now assume that the adversary choses $\R^s=\set{q}$ for all
$s>b$. 
Clearly, for process $q$, the sequence of states in the extension
starting from $C'$ and $C''$ is the same. Therefore, the two
runs are indistinguishable to $q$ also in this case.
\end{proof}

\begin{lemma}\label{lem:graph-seq}
Consider a round $r$ configuration $C$, then for any two round 
$r$ graphs $\G'$ and $\G''$, there is a $k$ such that we can find a
sequence of $k$ graphs $\G',\G_1,\dots\G_i\dots\G''$ each with a single
root component, where any two consecutive graphs differ only by at most one edge. 
Moreover, our construction guarantees that if $\G'$ and $\G''$ have the same root component $\R$ so do all $\G_i$.
\end{lemma}

\begin{proof} 
First, we consider two cases with respect to the root components $\R'$
and $\R''$: (a) $\R'\cap\R''=\emptyset$, (b)
$\R'\cap\R''\ne\emptyset$. Moreover, for the second part of the proof,
we also consider a special case of (b): (b') $\R'=\R''$.

For case (b) (and thus also for (b'), we consider $\G_1=\G'$.
For case (a), we construct $\G_1$ from $\G'$ as follows: Let $p'\in\R'$
and $p''\in\R''$, then $\G_1$ has the same edges as $\G'$ plus
$a=p''\edge{}p'$, thus $\R_1\supseteq\R'\cup\set{p''}$ (recall that
$p''$ must be reachable from $\R'$ already in $\G'$).
So, now we have that in both cases $\G'$ and $\G_1$ differ in at most
one edge. 
Moreover, there is a nonempty intersection between $\R_1$ and $\R''$.

In the first phase of our construction (which continues as long as $E''\setminus
E_i\ne\emptyset$), we construct $\G_{i+1}$ from $\G_{i}$, $i\ge1$, by
choosing one edge $e=(v\edge{}w)$ from $E''\setminus E_i$ and let
$\G_{i+1}$ have the same edges as $\G_i$ plus $e$.  Clearly, $\G_i$
and $\G_{i+1}$ differ in at most one edge. Moreover, when adding an
edge, we cannot add an additional root component, so as long as we add edges we
will have that $\G_{i+1}$ has a single root component $\R_{i+1}\supset
R'$.  

When we reach a point in our construction where $E''\setminus
E_i=\emptyset$, the first phase ends. As $\G_i$ now contains all the
edges in $\G''$, i.e., $E_i\supset E''$ and we have $\R_i\supset
R''$. In the second phase of the construction, we remove edges. To this
end, we choose one edge $e=(v\edge{}w)$ from $E_i\setminus E''$, and
construct $\G_{i+1}$ from $\G_i$ by removing $e$. Again we have to
show that there is only one root component. Since we never remove
an edge in $E''$, $\G_i$ always contains a directed path
from some $x\in\R''$ to both $v$ and $w$ that only uses edges in
$E''$. As $e\not\in E''$, this also holds for $\G_{i+1}$. Since there
is only one root component in $\G''$, this implies that there is only
one in $\G'$.

Let $\G_j$ be the last graph constructed in the first phase, and
$\G_k$ the last graph constructed in the second phase. 
It is easy to see that $E_k=E_j\setminus(E_j\setminus E'')$, 
which implies that $E_k=E''$ and hence $\G_k=E''$. This completes
the proof of the first part of the Lemma.

To see that the second part also holds, we consider case (b') 
in more detail and show by induction that $\R_{i+1}=\R_{i}=\R$. For
the base case we recall that $\G_1=G'$ and thus $\R_1=R'$. 
For the induction step, we consider first that the step involves
adding an edge $e=(v\edge{}w)$. Adding an edge can only modify the
root component when $v\not\in\R_i$ and $w\in\R_i$. Since such an edge $e$ is
not in $E''$ (as it has the same root component as $E'$), 
we cannot select it for addition, so the root component
does not change. If on the other hand the step from $\G_i$ to
$\G_{i+1}$ involves removing the edge $e=(v\edge{}w)$, then we only
need to consider the case where $v\in\R_i$. (If $v\not\in\R_i$, then
also $w\not\in\R_i$ so the root component cannot change by removing $e$.)
But since we never remove edges from $E''$, this implies that even
after removing $e$ there is still a path from $v$ to $w$, so the root
component cannot have changed. 
\end{proof}

\begin{theorem}
  Assume that Assumption~\ref{ass:too-short} is the only requirement for 
  the graph topologies.  Then consensus is impossible.
\end{theorem}
}
\onlyShort{
\begin{theorem}
  Assume that Assumption~\ref{ass:shorter} is not guaranteed in a
  system. 
  Then consensus is impossible.
\end{theorem}
}
\onlyLong{
\begin{proof}
We follow roughly along the lines of the proof of~\cite[Lemma
  3]{SWK09} and show per induction on the round number, that an
algorithm $A$ cannot reach a univalent configuration until round $r$. 

For the base case, we consider binary consensus only and argue similar to~\cite{FLP85} but make use of
our stronger Validity property: 
Let $C_x^0$ be the initial configuration, where the processes with
the $x$ smallest ids start with $1$ and all others with $0$. Clearly,
in $C_0^0$ all processes start with $0$ and in $C_n^0$ all start with
$1$, so the two configurations are $0$- and $1$-valent, respectively.
To see that for some $x$ $C_x^0$ must be bivalent, consider that this
is not the case, then there must be a $C_x^0$ that is $0$-valent while
$C_{x+1}^0$ is $1$-valent. But, these configurations differ only in
$p_{x+1}$, and so by Lemma~\ref{lem:neighbours-are-bivalent} they
cannot be univalent with different valency.

For the induction step we assume that there is a bivalent
configuration $C$ at the beginning of round $r-1$, and show that there is
at least one such configuration at the beginning of round $r$.
We proceed by contradiction and assume all configurations at the
beginning of round $r$ are univalent. Since $C$ is bivalent and all
configurations at the beginning of $r$ are univalent, there
must be two configurations $C'$ and $C''$ at the beginning of round
$r$ which have different valency. 
Clearly, $C'$ and $C''$ are reached from $C$ by two different 
round $r-1$ graphs $\G'=\li{\Pi,E'}$ and
$\G''=\li{\Pi,E''}$. Lemma~\ref{lem:graph-seq} shows that there is a
sequence of graphs that can be applied to both $C'$ and $C''$. 
Further, note that each pair of subsequent graphs in this sequence
differ only in one link $v\edge{}w$, so the resulting configurations
differ only in the state of $w$. Moreover, if the root component in
$\G'$ and $\G''$ is the same, all graphs in the sequence also have the
same root component. 
Since the valency of $C'$ and $C''$ was assumed to be
different, there must be two configurations $\overline{C}'$ and
$\overline{C}''$ in the sequence of configurations that have different 
valency and differ only in the state of one process, say $p$.
Applying Lemma~\ref{lem:neighbours-are-bivalent} to  $\overline{C}'$ and
$\overline{C}''$ again produces a contradiction, and so not all
successors of $C$ can be univalent. 
\end{proof}
}

\section{Conclusion} \label{sec:conclusion}
We introduced a framework for modeling dynamic networks with directed
communication links, and introduced a weak connectivity assumption that
makes consensus solvable. Without such assumptions,
consensus is trivially impossible in such systems as some processes can withhold their
input values until a wrong decision has been made. 
We presented an algorithm that achieves consensus under this assumption,
and showed several impossibility results and lower bounds that reveal that 
our algorithm is asymptotically optimal.

We have also explained that our framework is powerful enough to model
crash failures in the context of dynamic networks.
An important open question is thus how to handle Byzantine 
processes in \emph{dynamic} networks.
It is unclear whether connectivity assumptions that work in static 
networks (cf.\ \cite{Dolev82}) are also sufficient to solve consensus when the 
adversary has the additional power of manipulating the network topology.
\onlyShort{\bibliographystyle{splncs}}
\onlyLong{\bibliographystyle{plain}}
\bibliography{additional,lit}

\end{document}